\documentclass{article}
\pdfoutput=1

\usepackage{style}
\usepackage{shortcuts}
\usepackage[capitalize]{cleveref}
\usepackage[T1]{fontenc}

\def\eps{\varepsilon}
\def\epsilon{\varepsilon}

\title{Stronger 3-SUM Lower Bounds for Approximate Distance Oracles via Additive Combinatorics}
\author{%
	\parbox[t]{.29\textwidth}{\centering%
		Amir Abboud%
		\footnote{This work is part of the project CONJEXITY that has received funding from the European Research Council (ERC) under the European Union's Horizon Europe research and innovation programme (grant agreement No.~101078482). Supported by an Alon scholarship and a research grant from the Center for New Scientists at the Weizmann Institute of Science.}
		\\[.2ex]\small Weizmann Institute of Science}
	\and\parbox[t]{.33\textwidth}{\centering%
		Karl Bringmann%
		\footnote{This work is part of the project TIPEA that has received funding from the European Research Council (ERC) under the European Unions Horizon 2020 research and innovation programme (grant agreement No. 850979).}
		\\[.2ex]\small Saarland University
		\\\small Max Planck Institute for Informatics
		\\[-.6ex]\mbox{}}
	\and\parbox[t]{.29\textwidth}{\centering
		Nick Fischer%
		\footnote{Parts of this work were done while the author was at Saarland University and Max Planck Institute for Informatics. Partially supported by projects TIPEA and CONJEXITY as above.}
		\\[.2ex]\small Weizmann Institute of Science}}
\date{}

\begin{document}
\maketitle

\begin{abstract}
\noindent
The ``short cycle removal'' technique was recently introduced by Abboud, Bringmann, Khoury and Zamir~(STOC~'22) to prove fine-grained hardness of approximation. Its main technical result is that listing all triangles in an $n^{1/2}$-regular graph is $n^{2-o(1)}$-hard even when the number of short cycles is small; namely, when the number of $k$-cycles is $O(n^{k/2+\gamma})$ for $\gamma<1/2$. Its corollaries are based on the 3-SUM conjecture and their strength depends on $\gamma$, i.e. on how effectively the short cycles are removed. 

Abboud \emph{et al.}~achieve $\gamma\geq 1/4$ by applying structure versus randomness arguments on graphs. In this paper, we take a step back and apply conceptually similar arguments on the \emph{numbers} of the 3\=/SUM problem, from which the hardness of triangle listing is derived. Consequently, we achieve the best possible~$\gamma=0$ and the following lower bound corollaries under the 3-SUM conjecture:
\begin{itemize}
    \item\emph{Approximate distance oracles:} The seminal Thorup-Zwick distance oracles achieve stretch $2k\pm O(1)$ after preprocessing a graph in $O(m n^{1/k})$ time. For the same stretch, and assuming the query time is~$n^{o(1)}$ Abboud \emph{et al.} proved an \raisebox{0pt}[0pt][0pt]{$\Omega(m^{1+\frac{1}{12.7552 \cdot k}})$} lower bound on the preprocessing time; we improve it to \raisebox{0pt}[0pt][0pt]{$\Omega(m^{1+\frac1{2k}})$} which is only a factor~$2$ away from the upper bound. Additionally, we obtain tight bounds for stretch~$2+o(1)$ and~$3-\epsilon$ and higher lower bounds for dynamic shortest paths.
    \item\emph{Listing 4-cycles:} Abboud \emph{et al.} proved the first super-linear lower bound for listing all 4-cycles in a graph, ruling out $(m^{1.1927}+t)^{1+o(1)}$ time algorithms where $t$ is the number of 4-cycles. We settle the complexity of this basic problem by showing that the $\widetilde{O}(\min(m^{4/3},n^2) +t)$ upper bound is tight up to $n^{o(1)}$ factors.
\end{itemize}
Our results exploit a rich tool set from additive combinatorics, most notably the Balog-Szemerédi-Gowers theorem and Rusza's covering lemma. A key ingredient that may be of independent interest is a truly subquadratic algorithm for 3-SUM if one of the sets has small doubling.
\end{abstract}

\thispagestyle{empty}
\newpage
\setcounter{page}{1}

\section{Introduction} \label{sec:introduction}

An approximate distance oracle is an algorithm that preprocesses a graph efficiently and can then quickly return the distance between any given pair of nodes, up to a small error.
After being implicitly studied for some time \cite{matouvsek1996distortion,awerbuch1998near,cohen1998fast,DHZ00,cohen2001all}, Thorup and Zwick~\cite{TZ05} formally introduced the distance oracle problem in 2001 suggesting that it is perhaps the most natural formulation of the classical all-pairs shortest paths problem.
Distance oracles quickly rose to prominence and the techniques developed for them found deep connections to other popular topics such as sublinear algorithms, spanners, labelling schemes, routing schemes, and metric embeddings.

Distance oracles have been thoroughly investigated with the primary goal of understanding the best possible trade-off between the four main parameters: the multiplicative error factor (aka the \emph{stretch}), the query time, the space usage, and the preprocessing time;
see e.g. \cite{baswana2006approximate,mendel2007ramsey, baswana2008distance,BaswanaGS09,baswana2010faster,PR14,sommer2009distance,PRT12,WN12,WN13,Chechik14,Chechik15,sommer2016all,knudsen2017additive,AkavR20,RodittyT21,ChechikZ22} and the list is still growing.
They have also been studied from other perspectives, for example more efficient distance oracles for restricted classes of graphs were sought after (e.g.  \cite{CGMW19,LP21,LeWN21} for planar graphs), and their complexity in \emph{dynamic} graphs is of great interest (e.g. \cite{chechik2018near,GWN20,forster2021dynamic,DFNT22}).
Despite all this, perhaps the first question one might ask remains poorly understood: 

\begin{center}
{\em What is the best stretch $f(k)$ we can achieve if we insist on close-to-linear\\ preprocessing time $O(m n^{1/k})$ and almost-constant $n^{o(1)}$ query time?\footnote{Throughout we assume that graphs are undirected, unweighted and have $n$ nodes and $m$ edges.}}
\end{center}
\medskip

The seminal Thorup-Zwick oracle~\cite{TZ05} achieves stretch $2k-1$ after preprocessing a graph in $O(kmn^{1/k})$ time (it also achieves $O(k)$ query time and uses $O(n^{1+1/k})$ space).
Better trade-offs exist in the small $k$ regime of stretch below $3$ \cite{PR14,PRT12,BaswanaGS09,sommer2016all,knudsen2017additive,AkavR20,ChechikZ22}.
In dense enough graphs, the results are even better~\cite{baswana2006approximate,baswana2008distance,baswana2010faster,WN12}; in particular if \raisebox{0pt}[0pt][0pt]{$m= \Omega(n^{1+c/\sqrt{k}})$} Wulff-Nilsen~\cite{WN12} obtained \emph{linear} $O(m)$ preprocessing time.
However, in the setting of \emph{sparse} graphs and large $k$ (where the running time is close to linear), the Thorup-Zwick bound remains the state of the art.

Most of the existing lower bound techniques are incapable of answering the above question.
Incompressibility arguments \cite{bourgain1985lipschitz,matouvsek1996distortion,TZ05}, typically based on the girth conjecture, can show the optimality of the~$n^{1+1/k}$ space bound of Thorup and Zwick, but they cannot prove any lower bound higher than $m$. 
In the cell probe model, Sommer, Verbin, and Yu \cite{sommer2009distance} show that $m^{1+1/k}$ space (and therefore time) is required for stretch~$f(k)=O(k/t)$ if the query time is $t$; this lower bound is meaningless when the query time is super-constant and is far from the Thorup-Zwick upper bound even when $t$ is a small constant.
Finally, under a conjecture about the space complexity of Set Intersection, Pătraşcu, Roditty, and Thorup~\cite{PR14,PRT12} show $\Omega(mn^{\eps})$ lower bounds on the space complexity but their techniques only address stretch $3-\delta$; alas, they cannot prove that the error must grow above $3$ in the close-to-linear time regime.

At STOC '22, Abboud, Bringmann, Khoury, and Zamir \cite{AbboudBKZ22} introduced the \emph{short cycle removal} technique for hardness of approximation in fine-grained complexity and applied it to prove that the stretch must be~$f(k)>k/6.3776 \pm O(1)$, assuming the 3-SUM or APSP conjectures.
Thus, $f(k)$ must grow with $k$ and it is a linear function.
However, there is still a large gap in our understanding of this basic question; e.g.\ the optimal stretch for $O(mn^{0.1})$ preprocessing time could be anything between $21$ and $3$.
Whether the short cycle removal technique could achieve tight bounds was left as the main open question \cite{AbboudBKZ22}; the reasons for why this appears difficult are explained below.

\subsection{This Work: Optimal Short Cycle Removal}

In this work we take the short cycle removal technique to its limit and prove much higher and, in some cases, tight lower bounds using it.
Let us begin by introducing this technique.

\paragraph{The Starting Point}

Triangle finding problems are a common starting point for fine-grained hardness results.
The following \emph{all-edge} version is particularly interesting, since it is known to require $n^{2-o(1)}$ time in $\Theta(n^{1/2})$-regular graphs assuming either the 3-SUM conjecture \cite{Patrascu10,KopelowitzPP16} or the APSP conjecture \cite{VassilevskaX20}.

\begin{definition*}[All-Edges Triangle]
Given a tripartite graph $G=(V,E), V= X\cup Y \cup Z$ determine which edges in $E \cap (Y \times Z)$ are in at least one triangle.
\end{definition*}

Let us recall the popular 3-SUM conjecture that implies the hardness of All-Edges Triangle.

\begin{conjecture*}[3-SUM]
For any $\eps>0$, no $O(n^{2-\eps})$-time algorithm that can determine whether a given set~$A$ of~$n$ integers contains $a,b,c\in A$ such that $a+b+c=0$.
\end{conjecture*}

\paragraph{From All-Edges Triangle to Approximate Distance Oracles}
It is easy to reduce from All-Edges Triangle to distance oracles.
Construct a distance oracle for a new graph $G'$ that is obtained from $G$ by deleting all edges in $Y \times Z$.
To determine if an edge $(y,z) \in E(G) \cap (Y \times Z)$ is in a triangle, we query the oracle for the distance between $y$ and $z$ in $G'$: It must be exactly $2$ if $(y,z)$ is in a triangle in $G$, and it is at least $3$ otherwise.

To prove hardness for \emph{approximate} distance oracles we would want the distance in $G'$ to be much larger than $3$ if $(y,z)$ was not in a triangle in $G$. Now, the key observation is that a path of length $k-1$ in $G'$ implies that $(y,z)$ was in a $k$-cycle in $G$. In other words, if the edge $(y,z)$ is \emph{not} in a $k$-cycle in $G$ then a~$\frac k2$-approximation to the distance suffices for determining if $(y,z)$ participates in a triangle.

\paragraph{Short Cycle Removal} The basic idea of the \emph{short cycle removal} technique is to reduce \emph{the number of short cycles} in a graph without eliminating its triangles.
The goal is that the number of pairs $(y, z)$ that are in short cycles but not in triangles will be small, since such pairs incur a false positive in the above reduction.
The main tool towards this is to show that in subquadratic time the number of $k$-cycles can be reduced from the worst case $O(n^{k/2+1/2})$ to only $O(n^{k/2+\gamma})$ for $\gamma<1/2$ which is closer to the random case (where $\gamma=0$). The quality of the lower bounds obtained by this technique depends directly on the value of~$\gamma$ for which such a statement can be proved.

In \cite{AbboudBKZ22} the authors use the following structure versus randomness argument: If the graph has many $k$-cycles (more than the random case) then it must have a structure in the form of a \emph{dense piece} (a subgraph with disproportionately many edges). They use fast matrix multiplication to check for triangles that use the dense pieces and then remove them from the graph, reducing its number of $k$-cycles significantly and making it more random.

\begin{theorem*}[\cite{AbboudBKZ22}]
For any constants $\epsilon > 0, k_{\max} \geq 3$, there is no $\Order(n^{2-\epsilon})$-time algorithm for All-Edge Triangles in a $\Theta(n^{1/2})$-regular $n$-vertex graphs which contains at most~$\Order(n^{k/2+\gamma})$ $k$-cycles for all~\makebox{$3 \leq k \leq k_{\max}$} and for $\gamma=0.345+o(1)$, unless the 3-SUM and APSP conjectures fails.\footnote{In this theorem and the following ones, the restriction by $k_{\max}$ can be removed by~\cite[Lemma 5.11]{JinX23}.}
\end{theorem*}

The value of $\gamma$ that \cite{AbboudBKZ22} achieve depends on the fast matrix multiplication exponent $\omega<2.37188$ \cite{DWZ22,AlmanW21}, and even if $\omega=2$ they only get $\gamma=1/4$; going beyond this seems difficult.
The authors suggest an approach for getting $\gamma \to 0$ but there are three major barriers.
First, one needs to prove an interesting unproven combinatorial conjecture about the relationship between the number of cycles and the existence of dense subgraphs.
Second, one has to turn the proof into an efficient algorithm for finding the dense pieces.
And third, it is conceptually impossible to remove the dense pieces without using fast matrix multiplication, which means that one must first prove that $\omega \to 2$ before getting $\gamma\to 0$.

\paragraph{Optimal Short Cycle Removal}
In this paper we take a different approach: We look at the reduction from 3-SUM to All-Edges Triangle and ask: \emph{What structure in the 3-SUM instance causes the resulting graph to have too many $k$-cycles?}\footnote{In fact, we design a more transparent such reduction that could be of independent interest.}
The answer turns out to be related to the \emph{additive energy} of the 3-SUM instance, namely to the number of quadruples $a_1, a_2, a_3, a_4 \in A$ such that $a_1+a_2 = a_3 + a_4$. Thus, our goal changes from ``short cycle removal'' in graphs to ``\emph{energy reduction}'' on a set of numbers.
Surprisingly to us, the latter can be done much more effectively using machinery from additive combinatorics (overviewed in depth in Section~\ref{sec:overview}) such as the celebrated Balog-Szemerédi-Gowers theorem~\cite{BalogS94,Gowers01}.

Our main technical result is an optimal short cycle removal for All-Edges Triangle that is obtained via an optimal energy reduction for 3-SUM; along the way we prove new results for 3-SUM that are of independent interest (Theorems~\ref{thm:3sum-structured} and~\ref{lem:energy-reduction-complete} that are discussed in Section~\ref{sec:overview}). 
Notably, we achieve $\gamma=0$ even without assuming that $\omega = 2$.

\begin{theorem}[Optimal Short Cycle Removal]
For any constants $\epsilon > 0, k_{\max} \geq 3$, there is no $\Order(n^{2-\epsilon})$-time algorithm for All-Edges Triangle in a $\Theta(n^{1/2})$-regular $n$-vertex graph which contains at most~$\Order(n^{k/2})$ $k$-cycles for all $3 \leq k \leq k_{\max}$, unless the 3-SUM conjecture fails.
\end{theorem}

\subsection{New Lower Bounds for Distance Oracles}

Our main corollary is an improvement of the lower bound for distance oracles with $n^{o(1)}$ query time and close-to-linear $O(mn^{1/k})$ preprocessing time, from stretch $\geq k/6.3772$ to stretch $\geq k\pm O(1)$. 
This is only a factor $2$ away from the Thorup-Zwick upper bound. 
We find it interesting that the strongest known lower bound to our basic question about distance oracles involves tools from additive combinatorics.

\begin{restatable}[Hardness of Distance Oracles with Stretch $k$]{theorem}{thmhardnessdistanceoraclelarge} \label{thm:hardness-distance-oracles}
For any integer constant $k \geq 2$, there is no approximate distance oracle for sparse graphs with stretch $k$, preprocessing time $\widetilde\Order(m^{1+p})$ and query time~$\widetilde\Order(m^q)$ with $kp + (k+1)q < 1$, unless the 3-SUM conjecture fails.
\end{restatable}

Our lower bound is proved for sparse graphs where $m=O(n)$. Consequently, it cannot be bypassed even by $(\alpha,\beta)$-distance oracles that have an additive error of $\beta=n^{o(1)}$ in addition to a multiplicative stretch of $\alpha$.\footnote{This is because the additive error is insignificant in the sparse regime where we can subdivide edges.}
As in \cite{AbboudBKZ22}, our lower bound also holds for the \emph{offline} problem where we are given the queries before preprocessing. 
We also obtain a trade-off between the query and preprocessing time; e.g.\ if the query time is $O(n^{1/k})$ rather than $n^{o(1)}$ then the stretch is $k/2\pm O(1)$ rather than $k \pm O(1)$.

\paragraph{Tight Bounds?}
It may be disappointing that we did not get a tight lower bound despite optimizing the short cycle removal technique to its limit.
\emph{Is the short cycle removal technique inherently insufficient for proving a tight lower bound?}

The following three theorems indicate that our technique may well be the ``right'' one. 
This calls for revisiting the 20-year-old \emph{upper bounds} in the hope of closing the gap by improving the stretch from $2k-1$ to $k \pm O(1)$.
There is significant evidence that this may be around the corner. 
Better algorithms already exist in the regimes of dense graphs or when the stretch is small (some are very recent \cite{AkavR20,ChechikZ22}).
For large $k$, Roditty and Tov~\cite{RodittyT21} recently improved the $2k-1$ factor slightly to $2k-4$ while keeping the same space and query time as Thorup-Zwick (but not preprocessing time).
Moreover, in the closely related setting of \emph{graph spanners} where there is a similar trade-off saying that $2k-1$ stretch can be achieved with a subgraph on $n^{1+1/k}$ edges, it was shown by Parter~\cite{Parter14} that the stretch can be improved to $k$ for all pairs of nodes at distance $>1$ (see also \cite{EP04,BLP20}).
Alas, beating the Thorup-Zwick bound for general $k$ has been elusive; perhaps knowing that the gap from the lower bound is only $2$ (following this paper) will motivate the community to find better algorithms.
Such a result would not only be pleasing but it could also be useful in practice (see e.g. \cite{qi2013toward}).

\paragraph{The Small Stretch Regime}
Recall that the smallest stretch attainable by the Thorup-Zwick oracle is~$3$~(i.e. $k=2$), in which case their preprocessing time is $\Order(m\sqrt{n})$. 
Let us focus on the case of sparse graphs where~$m=O(n)$ and this time bound becomes $O(n^{1.5})$.
Subsequent work~\cite{BaswanaGS09,AkavR20,ChechikZ22} showed that interesting results can also be achieved for smaller stretch factors (if we allow constant additive error).
The most recent result by Chechik and Zhang~\cite{ChechikZ22} computes $(2+\eps,c_{\eps})$-approximations after $\widetilde{O}(m+n^{5/3+\eps})$ preprocessing time, for some constant $c_\eps$ that depends on $\eps$.

Using our optimal cycle removal, we prove that Thorup-Zwick is optimal in the following sense: If we want to improve the stretch to $3-\eps$ then the running time must grow polynomially to $m^{1.5+\Omega(\eps)}$.
In addition, we prove the optimality of the Chechik-Zhang algorithm in the sense that $m^{5/3-o(1)}$ time is required for stretch~$2+o(1)$.

\begin{restatable}[Hardness of Distance Oracles with Stretch $2 \leq \alpha < 3$]{theorem}{thmhardnessdistanceoraclesmall} \label{thm:hardness-distance-oracles-small}
For any \makebox{$2 \leq \alpha < 3$} and \makebox{$\epsilon > 0$}, in sparse graphs there is no distance oracle with stretch $\alpha$, query time $n^{\order(1)}$ and preprocessing time $\Order(m^{1+\frac{2}{1+\alpha}-\epsilon})$, unless the 3-SUM conjecture fails.
\end{restatable}

Recall that there are techniques besides short cycle removal that can prove lower bounds for stretch up to~$3$.
Indeed, the lower bounds of Pătraşcu, Roditty, and Thorup~\cite{PR14,PRT12} are similar to ours for stretches~\makebox{$2+o(1)$} and $3-\Omega(1)$, except that they are concerned with the preprocessing \emph{space} and not just time.
On the one hand, this makes their lower bounds stronger.
On the other hand, they need to rely on a strong conjecture about the \emph{space} versus query time trade-off of Set Intersection, rather than the 3-SUM conjecture that is simply about the \emph{time} complexity.
While Set Intersection is a common starting point for data structure lower bounds, the particular variant they use is not standard and was not used in any other paper to our knowledge.
Basing the same results also on one of the most popular conjectures in fine-grained complexity is desirable.
In any case, the more important message of Theorem~\ref{thm:hardness-distance-oracles-small} is to show that our techniques \emph{can} prove tight bounds.

\paragraph{Dynamic Graphs}
Extending our basic question to the dynamic setting we seek the optimal stretch for a distance oracle that achieves $n^{o(1)}$ query time and $O(n^{1/k})$ time for \emph{updates} that add or remove an edge. 
In this case, we give a more efficient reduction from All-Edges Triangle (inspired by the reduction of Abboud and Vassilevska Williams~\cite{AV14} to dynamic matching) and prove that the stretch must be $2k\pm O(1)$.

\begin{restatable}[Hardness of Dynamic Distance Oracles]{theorem}{thmhardnessdistanceoracledynamic} \label{thm:hardness-distance-oracles-dynamic}
For any integer constant $k \geq 2$, there is no dynamic approximate distance oracle with stretch $2k-1$, update time $\Order(m^u)$ and query time $\Order(m^q)$ with $ku + (k+1)q < 1$, unless the 3-SUM conjecture fails.
\end{restatable}

This lower bound would be tight if the Thorup-Zwick bound extends to the dynamic setting. This was indeed accomplished by Chechik~\cite{chechik2018near} (see also~\cite{GWN20,DFNT22}) in the \emph{decremental} setting where only edge deletions are allowed, but not yet in the fully dynamic case (the best bounds appear in \cite{forster2021dynamic}).

\subsection{A Tight Lower Bound for 4-Cycle Listing}

Finding $4$-cycles in a graph is one of the simplest non-trivial cases of the classical \emph{Subgraph Isomorphism} problem.
The longstanding upper bound for testing $4$-cycle freeness is $O(\min(n^2,m^{4/3}))$~\cite{AYZ97,YZ97}.
It is conjectured that no $O(n^{2-\eps})$ algorithm exists; proving this under one of the more popular conjectures of fine-grained complexity has been a well-known open question.
In fact, the $4$-cycle problem is infamous for eluding even any super-linear lower bound via the standard reduction techniques; \cite{AbboudBKZ22} highlight this problem as encapsulating the challenge in proving hardness of approximation results for distance oracles.

In the \emph{listing} version we are asked to output all $4$-cycles in the graph.
Such problems are well-studied and are closely related to the enumeration of query answers in databases.
It is known that \emph{all} cycles in a graph can be listed in linear time $O(m+t)$ where $t$ is the output size \cite{birmele2013optimal}.
But for a fixed length $k$, the listing $k$-cycles problem is not as easy.
In a landmark result in fine-grained complexity, that implied the aforementioned 3-SUM-hardness for All-Edge Triangle, Pătraşcu~\cite{Patrascu10} proved an essentially tight lower bound for \emph{triangle} listing (see \cite{KopelowitzPP16,bjorklund2014listing}).
The first and only super-linear lower bound for $4$-cycle listing, however, came only a decade later via the short cycle removal technique \cite{AbboudBKZ22}.
We improve their lower bound from $(m^{1.1927}+t)^{1-o(1)}$ to a \emph{completely tight} lower bound matching the $O(\min(n^2,m^{4/3})+t)$ upper bound.\footnote{The $O(n^2+t)$ upper bound is simple: Create an array of size $n^2$. For each node $x$ and all pairs of neighbors $u,v \in N(x)$ store $x$ in the $(u,v)$ entry of the array. If we access an entry that already contains nodes $y_1,\ldots,y_k$ we output the $4$-cycles $(x,u,y_i,v)$ for all $i$. The time is $n^2$ plus the number of $4$-cycles because each time we access an entry (except for the first time) we output at least one $4$-cycle. The $O(m^{4/3}+t)$ algorithm is more involved \cite{AKLS22}.}
It is remarkable that a resolution of the complexity of this crisp problem is attained via additive combinatorics.

\begin{restatable}[Hardness of Listing 4-Cycles]{theorem}{thmhardnessfourcycle} \label{thm:hardness-4cycle}
For any $\epsilon > 0$, there is no algorithm listing all 4-cycles in time $\widetilde\Order(n^{2-\epsilon} + t)$ or in time $\widetilde\Order(m^{4/3-\epsilon} + t)$ (where $t$ is the number of 4-cycles), unless the 3-SUM conjecture fails.
\end{restatable}

\subsection{Independent Work by Jin and Xu}
Independently to our research and at the same time, Jin and Xu have discovered almost the same results as the ones presented in this paper~\cite{JinX23}. While very similar on a high level, our papers differ in their respective focus points: We focus in depth on the hardness of distance oracles in several regimes, whereas Jin and Xu focus more broadly on consequences throughout fine-grained complexity and in particular establish related lower bounds against \emph{4-linear degeneracy testing}. On the technical side our works are also very similar on a high level, but differ in the details.

\subsection{Outline}
We start with some preliminaries in \cref{sec:preliminaries}. In \cref{sec:overview} we give a technical overview of our results, with some background knowledge from additive combinatorics postponed to \cref{sec:additive-combinatorics}. In \cref{sec:3sum-structured,sec:3sum-energy-reduction,sec:3sum-to-triangle} we give the key steps of our lower bound related to 3-SUM, and in \cref{sec:distance-oracles,sec:4cycle} we provide the specific lower bounds for distance oracles and 4-cycle listing. Finally, in \cref{sec:sparse-sumset,sec:bsg} we provide some missing proofs.
\section{Preliminaries} \label{sec:preliminaries}
We set $[n] = \set{1, \dots, n}$, and write $\widetilde\Order(n) = n (\log n)^{\Order(1)}$ to suppress polylogarithmic factors. Throughout, all algorithms are randomized and succeed \emph{with high probability}, that is, with error probability~$1/n^c$ for an arbitrarily large constant $c$.

\paragraph{Sumset Notation}
We often fix a group $G = \Int$ (the integers) or $G = \Field_p^d$ (a vector space over the finite field of prime order $p$). For sets $A, B \subseteq G$, we define the sumset notation~$A + B = \set{a + b : a \in A, b \in B}$ and $A - B = \set{a - b : a \in A, b \in B}$. Occasionally we write~$kA$ as the iterated sumset $A + \dots + A$ (with $k$ terms).

\paragraph{3-SUM} The \emph{monochromatic 3-SUM problem} is to decide whether in a given set $A$, there are $a, b, c \in A$ (not necessarily distinct) such that $a + b + c = 0$. We say that the instance has size $n = |A|$. The \emph{trichromatic 3-SUM problem} is to decide whether in three given sets $A, B, C$, there are $a \in A, b \in B, c \in C$ such that~$a + b + c = 0$. We say that the instance has size $n = |A| + |B| + |C|$. Both variants are equivalent in terms of subquadratic algorithms. We typically work under the well-established assumption that 3-SUM requires quadratic time~\cite{GajentaanO95}.

\paragraph{Graphs}
In this paper all graphs are undirected and unweighted. The \emph{distance $d(u, v)$} of two vertices is the length of the shortest path from $u$ to $v$. We say that a graph is $\Theta(r)$-regular if there are constants~$0 < c_1 < c_2$ such that every vertex has degree $\deg(v)$ satisfying~\makebox{$c_1 r \leq \deg(v) \leq c_2 r$}.
\section{Technical Overview} \label{sec:overview}
In this section, we give a high-level overview of our results.

\subsection{Hardness Reductions from Triangle Listing Instances with Few Short Cycles}
We start with the motivating observation that, if we could assume hardness of triangle listing in \emph{random-like} graphs, we could rather easily conclude tight hardness of 4-cycle listing. Under the same assumption and with some more work, we can also show the promised hardness of distance oracles. More specifically, assume that it is $n^{2-\order(1)}$-hard to list $\Order(n^{3/2})$ triangles in a $\Theta(n^{1/2})$-regular graph which contains at most $\Order(n^2)$ 4-cycles---this is indeed the number of 4-cycles we expect in a random $\Theta(n^{1/2})$-regular graph.

\paragraph{Hardness of Listing 4-Cycles}
First, by a simple subsampling trick we can reduce the number of 4-cycles in the given triangle instance a tiny bit further: We randomly split the vertex set into $n^\delta$ many groups and list all triangles in each triple of groups. In this way we incur an overhead of $n^{3\delta}$ to the running time. However, we have reduced the total number of 4-cycles (across all triples of groups) to $\Order(n^{2-\delta})$. Indeed, each 4-cycle falls into a fixed triple of groups only with probability $n^{-4\delta}$, and thus the total number of 4-cycles is~$n^{3\delta} \cdot \Order(n^{2-4\delta}) = \Order(n^{2-\delta})$.

We follow a natural approach on the smaller instances~$G = (V, E)$ (i.e., for each triple of groups): We create a new graph consisting of four copies~$V_1, V_2, V_3, V_4$ of $V$ (i.e., each vertex~$v \in V$ now has four copies~$v_1, v_2, v_3, v_4$). We add all edges from $E$ between the parts~$V_1$ and~$V_2$, between~$V_2$ and~$V_3$ and between~$V_3$ and~$V_4$. Finally, we connect all matching vertices in $V_1$ and $V_4$ (i.e., for all $v \in V$ we add the edge $(v_1, v_4)$).

With this construction, each triangle $(u, v, w)$ in the original graph can now be found as a 4-cycle $(u_1, v_2, w_3, u_4)$. However: There might be many more 4-cycles which do not correspond to triangles in the original instance, e.g., 4-cycles which only zigzag between the vertex parts $V_1$ and $V_2$. These 4-cycles must be part of the original graph though, and thus the total number of 4-cycles in the instance is bounded by $\Order(n^{2-4\delta})$. It follows that if there is an algorithm listing all $t$ 4-cycles in a graph in time $\Order(n^{2-\epsilon} + t)$, then we could list all triangles in time $\Order(n^{2-\epsilon} + n^{2-4\delta})$. The total time across all triples of groups is bounded by~$\Order(n^{2-\epsilon+3\delta} + n^{2-\delta})$, which is subquadratic by setting $\delta > 0$ small enough.

\paragraph{Hardness of Approximate Distance Oracles}
In a similar spirit we derive hardness results for distance oracles. We achieve results for several settings (see \cref{thm:hardness-distance-oracles,thm:hardness-distance-oracles-small,thm:hardness-distance-oracles-dynamic}), but in this overview we will only focus on the simplest version to get the idea across. We demonstrate how to rule out distance oracles with stretch~$k$, constant query time and preprocessing time~$\Order(m^{1+\frac1{2k+1}-\epsilon})$. (This is a weaker bound than in \cref{thm:hardness-distance-oracles}, where we even rule distance oracles with preprocessing time $\Order(m^{1+\frac1k-\epsilon})$).

We again start from an instance $G$ of listing $\Order(n^{3/2})$ in a $\Theta(n^{1/2})$-regular $n$-vertex graph, and assume that the graph contains at most $\Order(n^{k/2})$ $k$-cycles, for all $k$. Without loss of generality assume that the instance is a tripartite graph $G = (X, Y, Z, E)$ with vertex parts~$X, Y, Z$. We will uniformly subsample all vertex parts with some rate $\rho$ to obtain a smaller graph $G'$ with vertices $X' \subseteq X, Y' \subseteq Y, Z' \subseteq Z$. This graph is~$\Theta(\rho n^{1/2})$-regular, has~$\Order(\rho n)$ vertices and has $\Order(\rho^2 n^{3/2})$ edges. Most interestingly though, the number of $k$-cycles in~$G'$ is at most $\Order(\rho^k n^{k/2})$, as every $k$-cycle survives the subsampling only with probability~$\rho^k$.

We will now use the distance oracle to efficiently list all triangles in $G'$. To this end, let~$G''$ be a duplicate of $G'$ where we delete the edges between $X'$ and $Z'$; any pair of vertices $(x, z) \in X' \times Z'$ which was part of a triangle has distance $d(x, z) \leq 2$ in $G''$. We preprocess~$G''$ with the distance oracle and query each pair~$(x, z) \in (X' \times Z') \cap E$ to obtain distance estimates $d(x, z) \leq \widetilde d(x, z) \leq k \cdot d(x, z)$. We say that a pair~$(x, z)$ is a \emph{candidate} if its distance estimate is $\widetilde d(x, z) \leq 2k$. The idea is that only the candidate pairs can possibly be part of a triangle---as all other pairs must have distance more than $2$ in $G''$. However, note that among the candidate pairs there may be many pairs which do not form a triangle. Our listing algorithm now enumerates all candidate pairs $(x, z)$ and all neighbors $y \in Y'$ of $x$ and tests whether $(x, y, z)$ forms a triangle. It should be clear that the algorithm cannot miss any triangle in $G'$. And by repeating the subsampling $\widetilde\Order(\rho^{-3})$ times, with good probability every triangle in $G$ occurs in at least one instance $G'$ and will therefore eventually be detected.

The running time is dominated by two major contributions: The preprocessing time of the distance oracle and the enumeration step (for this setting of parameters the query time can be ignored). The total preprocessing time across all $\widetilde\Order(\rho^{-3})$ repetitions is bounded by
\begin{equation} \label{eq:distance-oracles-preprocessing}
    \widetilde\Order(\rho^{-3} \cdot (\rho^2 n^{3/2})^{1+\frac{1}{2k+1}-\epsilon}).
\end{equation}

Next we deal with the contribution of the enumeration step. The key in the analysis is to get a good bound on the number of candidate pairs $(x, z)$. Observe that as any candidate pair~$(x, z)$ has distance $d(x, z) \leq 2k$ in $G''$, it must be part of a cycle of length at most~$2k+1$ in $G'$. We can thus control the number of candidate pairs by controlling the number of cycles in~$G'$---as argued before, there are most \raisebox{0pt}[0pt][0pt]{$\Order(\rho^{2k+1} n^{\frac{2k+1}2})$} cycles of length at most $2k+1$. Dealing with a single candidate pair takes time $\Order(\rho n^{1/2})$ (to list all neighbors~$y$ of~$x$), and therefore the total running time of the enumeration step is bounded by
\begin{equation} \label{eq:distance-oracles-enumeration}
    \Order(\rho^{-3} \cdot \rho^{2k+1} n^{\frac{2k+1}2} \cdot \rho n^{1/2}) = \Order(\rho^{2k-1} n^{k+1})
\end{equation}

By optimizing $\rho$ in both contributions~\eqref{eq:distance-oracles-preprocessing} and~\eqref{eq:distance-oracles-enumeration}, we find that the running is indeed subquadratic (for~\raisebox{0pt}[0pt][0pt]{$\rho = n^{-\frac{k-1}{2k-1}-\delta}$} and some tiny $\delta > 0$).  This completes the proof outline of the weaker lower bound. For the improved lower bound from \cref{thm:hardness-distance-oracles}, we find better trade-off between the size of the preprocessed graph and the number of queries to the distance oracle.

\paragraph{Revisiting Hardness of Listing Triangles}
The main message is that if miraculously the given triangle instance contains few 4-cycles, then we would obtain interesting hardness results. We therefore investigated whether this variant of triangle listing is conditionally hard, and managed to prove the desired result:

\begin{restatable}[Hardness of Triangle Listing]{theorem}{thmthreesumtotriangle} \label{lem:3sum-to-triangle}
For any constants $\epsilon > 0, k_{\max} \geq 3$, there is no $\Order(n^{2-\epsilon})$-time algorithm listing all triangles in a $\Theta(n^{1/2})$-regular $n$-vertex graph which contains at most~$\Order(n^{k/2})$ $k$-cycles for all $3 \leq k \leq k_{\max}$, unless the 3-SUM conjecture fails.
\end{restatable}

There are known lower bounds against listing $\Order(n^{3/2})$ triangles in $\Theta(n^{1/2})$-regular graphs (without the assumption that the graph has few short cycles) under the 3-SUM conjecture by Pătraşcu~\cite{Patrascu10} with refinements by Kopelowitz, Pettie and Porat~\cite{KopelowitzPP16} and under the All-Pairs Shortest Paths conjecture by Vassilevska Williams and Xu~\cite{VassilevskaX20}. We specifically focused on the 3-SUM hardness and as a first step significantly simplified the known reduction (see \cref{sec:3sum-to-triangle}). We then raised the question: In this reduction from 3-SUM to triangle listing, \emph{what makes the constructed triangle instance have many 4-cycles?}

It turns out that the number of 4-cycles in the triangle instance is controlled by the number of solutions to the equation $a_1 + a_2 = a_3 + a_4$, where $a_1, a_2, a_3, a_4 \in A$, in the 3-SUM instance $A$. In the additive combinatorics literature this quantity is commonly referred to as the \emph{additive energy $E(A)$} of~$A$. Note that $E(A)$ ranges from $n^2$ (as there are at least $n^2$ trivial solutions with~$a_1 = a_3$ and $a_2 = a_4$) to $n^3$ (as any fixed values $a_1, a_2, a_3$ uniquely determine $a_4$). A~set with additive energy close to $n^2$ is considered \emph{unstructured}---for instance a random set has expected energy~$\Order(n^2)$. A set with energy close to $n^3$ is considered structured---examples include intervals and arithmetic progressions.

In summary: To obtain a triangle listing instance containing few 4-cycles, we have to start from a 3-SUM instance with very small additive energy $E(A)$ (in \cref{sec:3sum-to-triangle} we prove this statement in detail).

\subsection{Energy Reduction for 3-SUM}
We manage to show a self-reduction for 3-SUM which reduces the energy down to $\Order(n^2)$. We will refer to this type of reduction as an \emph{energy reduction} for 3-SUM. Our outline for the energy reduction is as follows: First, we reduce the additive energy by a tiny bit, say to $n^{2.9999}$, using several tools from additive combinatorics. Second, we apply a randomized 3-SUM self-reduction (which can be seen as an efficient way of \emph{subsampling} the instance) to amplify the tiny improvement to an arbitrarily large improvement. We will now describe both steps in more detail.

\paragraph{First Step: Energy Reduction via Additive Combinatorics}
The precise result we obtain in this step is as follows. Here, and in fact throughout the whole paper, we will set $K = n^{0.0001}$. Moreover, throughout let $G$ be either $G = \Int$ or $G = \Field_p^d$.

\begin{restatable}[Energy Reduction via Additive Combinatorics]{lemma}{lemenergyreductionadditivecombinatorics} \label{lem:energy-reduction-additive-combinatorics}
Let $K \geq 1$. There is a fine-grained reduction from a 3-SUM instance $A$ of size $n$ to an equivalent 3-SUM instance~$A^* \subseteq A$, where~$E(A^*) \leq |A^*|^3 / K$. The reduction runs in time $\widetilde\Order(K^{314} n^{7/4})$.
\end{restatable}

A key ingredient for this step is the seminal Balog-Szemerédi-Gowers theorem (in short: the BSG theorem). Intuitively, the theorem states that every set $A$ with large additive energy $E(A)$ must contain a large subset $A'$ which behaves like an interval or an arithmetic progression in the sense that its sumset $A + A = \set{a_1 + a_2 : a_1, a_2 \in A}$ has very small size (we also say that $A$ has small \emph{doubling}). The theorem can be formally stated as follows:

\begin{restatable}[Balog-Szemer\'edi-Gowers]{theorem}{thmbsg} \label{thm:bsg}
Let $A \subseteq G$. If $E(A) \geq |A|^3 / K$, then there is a subset $A' \subseteq A$ such that
\begin{itemize}
\item $|A'| \geq \Omega(K^{-2} |A|)$, and
\item $|A' + A'| \leq \Order(K^{24} |A|)$.
\end{itemize}
Moreover, we can compute $A'$ in time $\widetilde\Order(K^{12} |A|)$ by a randomized algorithm.
\end{restatable}
The existential part of the theorem (without the claimed running time bounds) was originally proved by Balog and Szemerédi~\cite{BalogS94} and Gowers~\cite{Gowers01}. The efficient algorithm to compute $A'$ was later devised by Chan and Lewenstein~\cite{ChanL15} based on a proof of the BSG theorem which was independently discovered by Balog~\cite{Balog07} and Sudakov, Szemer\'edi and Vu~\cite{SudakovSV05}.

The BSG theorem suggests the following algorithmic idea: As long as $A$ has large additive energy, apply the BSG theorem to extract a highly structured subset $A' \subseteq A$, and efficiently solve 3-SUM on that set. More specifically, we have to solve the trichromatic 3-SUM instance~$(A', A, A)$, where we can assume that~$|A' + A'| \leq \Order(K^{24} |A|) \leq \Order(K^{26} |A'|)$. Indeed, either there exists a solution contained in $A \setminus A'$ in which case we can simply discard~$A'$, or part of the 3-SUM solution is contained in $A'$ in which case this will be a valid solution in~$(A', A, A)$. One can prove that after at most $\widetilde\Order(K)$ extractions, we have either found a 3-SUM solution or the remaining set has small additive energy as required.

It remains to solve the 3-SUM instances $(A', A, A)$. There are some known results about structured 3-SUM instances: For instance, using sparse convolution algorithms we can solve 3-SUM instances $(A, B, C)$ in subquadratic time whenever~$A + B$ has subquadratic size. Another result by Chan and Lewenstein~\cite{ChanL15} is that 3-SUM admits subquadratic-time algorithms whenever one of the sets is \emph{clustered}, that is, if it can be covered by a subquadratic number of size-$n$ intervals. Unfortunately, neither of these algorithms can be applied in our context and to the best of our knowledge no algorithm is known for the case when one of the input sets has small doubling. It is one of our key technical contributions to design an algorithm for this problem:

\begin{restatable}[3-SUM for Structured Inputs]{theorem}{thmthreesumstructured} \label{thm:3sum-structured}
Let $(A, B, C)$ be a 3-SUM instance of size $n$ with $A, B, C \subseteq G$ and $|A + A| \leq K |A|$. Then we can solve $(A, B, C)$ in time $\widetilde\Order(K^{12} n^{7/4})$.
\end{restatable}

We omit the description of this algorithm for now and continue with the energy reduction. Later in the overview, in \cref{sec:overview:sec:3sum-structured}, we give the main ideas and in \cref{sec:3sum-structured} we provide the detailed proof of \cref{thm:3sum-structured}.

\paragraph{Second Step: Amplification via Hashing}
In the previous step we have reduced a worst-case 3-SUM instance to another instance with a tiny improvement in additive energy. In this step, we will amplify this improvement by means of the following reduction:

\begin{restatable}[Energy Reduction via Hashing]{lemma}{lemenergyreductionsubsampling} \label{lem:energy-reduction-subsampling}
Let $K \geq 1$. There is a fine-grained reduction from a 3-SUM instance $A$ with $E(A) \leq |A|^3 / K$ to $g = \Order(|A|^2 / K^2)$ 3-SUM instances $A_1, \dots, A_g$ of size $\Order(K)$ and with expected energy $\Ex(E(A_i)) \leq \Order(K^2)$. The reduction runs in time~$\widetilde\Order(|A|^2 / K)$.
\end{restatable}

The rough idea behind \cref{lem:energy-reduction-subsampling} is to create many randomly subsampled instances from~$A$. An efficient way to implement such a self-reduction is to not subsample $A$ uniformly, but instead make use of linear hashing. This general idea is not new and has appeared several times before in the context of 3-SUM~\cite{BaranDP08,Patrascu10,KopelowitzPP16,ChanH20}, but we have to pay closer attention than usual in order to analyze the additive energy.

We describe the idea in a simplified way, glimpsing over several problems: Sample a linear hash $h$ to~$m$ buckets and create the instance $B = \set{a \in A^* : h(a) = 0}$. Here, \emph{linear} means that $h$ satisfies the condition~$h(a) + h(b) = h(a + b)$ for all inputs $a, b$. What is the probability that a fixed 3-SUM solution~$a + b + c = 0$ survives? The probability is at least~$1/m^2$, since~$1/m^2$ is the probability that $h(a) = h(b) = 0$, which entails that also~$h(c) = 0$ by the linearity of the hash function. This means that we have to repeat this reduction $\widetilde\Omega(m^2)$ times until a 3-SUM solution survives.

In contrast, what is the probability that a solution to the equation~\makebox{$a_1 + a_2 = a_3 + a_4$} survives? By the same linearity argument, we can only use the randomness for three of the four variables as the hash value of the remaining variable is fixed. We therefore expect each solution to survive with probability $1/m^3$. Since this probability is smaller by a factor~$m$ compared to the survival probability of a 3-SUM solution, only a $1/m$-fraction of solutions~$a_1 + a_2 = a_3 + a_4$ survives and appears in one the small instances. In particular, by setting~$m = n / K$ we create $n^2 / K^2$ instances of size $n/m = K$ and with additive energy bounded by~$E(A) / m^3 \leq (n/m)^3 / K = K^2$.

However, there is a serious issue with this approach: In order to argue that each solution to the equation~$a_1 + a_2 = a_3 + a_4$ survives with probability at most $1/m^3$ we have assumed that three elements, say,~$a_1$,~$a_2$ and~$a_3$, are hashed \emph{independently}. Unfortunately there are no hash functions which are linear and 3-wise independent at the same time. We will ignore this issue for now, and explain later in \cref{sec:overview:sec:hashing} how to overcome this challenge.

By combining both steps of the energy reduction, we obtain the following theorem:

\begin{restatable}[Energy Reduction]{theorem}{thmenergyreduction} \label{lem:energy-reduction-complete}
For any $\epsilon, \delta > 0$, there is no $\Order(n^{2-\epsilon})$-time algorithm solving the 3-SUM problem on instances $A$ with size $n$ and additive energy $E(A) \leq \Order(|A|^{2+\delta})$, unless the 3-SUM conjecture fails.
\end{restatable}

\paragraph{Comparison to Abboud, Bringmann, Khoury and Zamir~\cite{AbboudBKZ22}}
We remark that our approach for an energy reduction is conceptually similar to the work of Abboud, Bringmann, Khoury and Zamir~\cite{AbboudBKZ22}: Their goal was also to reduce the number of 4-cycles in a triangle instance. They achieved this by first reducing the number of 4-cycles by a little bit (by identifying and removing dense pieces in the graph, which contain many 4-cycles), and then subsample the remaining instance to amplify the 4-cycle reduction. In contrast to our setting, working on the triangle instances directly has the disadvantage that sparse triangle problems are not known to admit efficient self-reductions. As a result, their subsampling step is lossy and leads to non-matching lower bounds.

\bigskip
This completes the description of the energy reduction. In the following subsections we describe what we left out in the previous overview---how to efficiently solve 3-SUM for structured inputs and how to deal with the hashing issue.

\subsection{3-SUM for Structured Inputs} \label{sec:overview:sec:3sum-structured}
In this section we describe a subquadratic-time algorithm for 3-SUM instances $(A, B, C)$ in which the set $A$ has doubling $|A + A| \leq K |A|$. We first describe a simple toy algorithm to build some intuition.

\paragraph{Warm-Up: \boldmath$A$ Is Contained in an Interval}
We give a simple algorithm that works whenever~$A$ is contained in a small interval, say $I = [10n]$ (this is indeed an example of a set with small doubling). Our approach is to \emph{cover} $B$ and $C$ by translates of $I$. That is, we split $B$ into a collection of disjoint subsets~$B_1, \dots, B_\ell$ each of which is obtained by intersecting $B$ with a translate of~$I$. Note that we need at most $|B|$ translates to cover the full set $B$. We similarly cover~$C$ by disjoint subsets $C_1, \dots, C_m$. The insight is that $A + B_i$ is contained in an interval of size $20n$. Therefore if there is a 3-SUM solution~$(a, b, c) \in A \times B \times C$ with~$b \in B_i$, there are at most three sets $C_j$ which could possibly contain $c$. Calling a pair~$(i, j)$ \emph{relevant} if there could possibly be a 3-SUM solution in $A \times B_i \times C_j$, we have argued that the number of relevant pairs is at most $\Order(n)$.

We iterate over all relevant pairs $(i, j)$, and apply a heavy-light approach: If both sets~$B_i$ and $C_j$ have size at most $n^{1/3}$, then we brute-force over all pairs $(b, c) \in B_i \times C_j$ and test whether they constitute a 3-SUM solution. Otherwise, we compute $B_i + C_j$ using FFT, and test for each element in the sumset whether it is part of a 3-SUM solution with $A$. The total time of the light case is bounded by $\Order(n)$ (the number of relevant pairs) times $\Order(n^{2/3})$ (the number of pairs $(b, c)$ we explicitly test). The total time of the heavy case is bounded by~$\Order(n^{2/3})$ (there can be at most that many relevant pairs $i, j$ for which either $B_i$ or $C_j$ has size larger than $n^{1/3}$) times $\widetilde\Order(n)$ (running FFT on sets of universe size $20n$). The total time is $\Order(n^{5/3})$, which is subquadratic.

The take-away message is that when we know that $A$ is contained in a small interval, we can benefit from the structure by pruning the search space in $B \times C$ (i.e., we do not compare every element in $B$ to every element in $C$). The question is: \emph{What is the appropriate generalization of an interval?}

\paragraph{Full Algorithm: \boldmath$A$ Is Contained in an Approximate Group}
For us, the appropriate generalization are \emph{approximate groups.} A set~$H$ is a $K$-approximate group if $H = -H$ and~$H + H$ can be covered by at most~$K$ translates of $H$. The key ingredient to our algorithm is yet another result from additive combinatorics: Ruzsa's covering lemma. More specifically, we exploit the following consequence of Ruzsa's covering lemma which states that any set with small doubling can be covered by a small approximate group.

\begin{restatable}[Covering by Approximate Groups]{lemma}{lemcoverapproximategroup} \label{lem:cover-approximate-group}
Let $A \subseteq G$ be a set with $|A + A| \leq K |A|$. Then there is a set $H \subseteq G$ with the following properties:
\begin{itemize}
\item $|H| \leq K^2 |A|$,
\item $H$ is a $K^5$-approximate group, that is, there is some set $X \subseteq G$ of size $|X| \leq K^5$ such that $H = -H$ and $H + H \subseteq H + X$, and
\item there is some $a_0 \in A$ such that $A - a_0 \subseteq H$.
\end{itemize}
Moreover, we can compute $H$, $X$ and $a_0$ in time $\widetilde\Order(K^{12} |A|)$.
\end{restatable}

The existential result is well-known in additive combinatorics (see for instance the book by Tao and Vu~\cite{TaoV06}), but for our purposes it is also important to have an efficient algorithm to compute $H$. We derive an algorithm based on computing sparse convolutions, see the proof in \cref{sec:additive-combinatorics}.

For our 3-SUM algorithm, thanks to \cref{lem:cover-approximate-group} we can assume that $A$ is contained in (a translate of) a small approximate group $H$. We mimic the warm-up algorithm with the same idea: Cover~$B$ and~$C$ by translates of $H$, say $B_1, \dots, B_\ell$ and $C_1, \dots, C_m$. For each set~$B_i$, there are only few sets $C_j$ which are candidates to contain a 3-SUM solution, namely at most~$K^5$ many. Therefore, we can apply a similar heavy-light approach as outlined before, and either enumerate all pairs in $B_i \times C_j$ if both sets are sparse, or efficiently compute~$B_i + C_j$ using a sparse sumset algorithm (in place of FFT) if one of the sets is dense.

An additional difficulty is that we cannot simply cover $B$ and $C$ by translates of $H$ in linear time. (For intervals this is easy, but we have no information about $H$ other than that is an approximate group.) We therefore sample a set $S$ of \emph{random} shifts and attempt to cover $B$ by the sets~$B \cap (H + s)$ for $s \in S$ (similarly for $C$). However, computing the sets~$B \cap (H + s)$ is not easy (in fact, this is again an instance of 3-SUM). We deal with this new obstacle by combining the above algorithm with a universe reduction to a universe of subquadratic size. The detailed proof can be found in \cref{sec:3sum-structured}.

We remark that we have not attempted to improve the dependence on $K$ as it is immaterial for our reduction. It is likely possible to drastically reduce the $K$ term in the running time of \cref{thm:3sum-structured}.

\subsection{Hashing---Linear and Independent?} \label{sec:overview:sec:hashing}
A major technical issue that we are facing in the energy reduction (and in fact also in the reduction from 3-SUM to listing triangles) is that we need hash functions which are both \emph{linear} and sufficiently \emph{independent}. More specifically, recall that we hash a given 3\=/SUM instance $A$ to a smaller instance $B = \set{a \in A : h(a) = 0}$ hoping that thereby the number of solutions to the equation $a_1 + a_2 = a_3 + a_4$ reduces by a factor of $1/m^3$, where $m$ is the number of buckets $h$ hashes to. By the reasons outlined before, the hash function \emph{must} be linear. Unfortunately, no matter what construction we use, for a linear hash function the random variables $h(a_1), h(a_2), h(a_3)$ cannot always be independent. For example, $h(1)$, $h(2)$ and $h(5)$ cannot be independent, since the latter can be expressed as~$h(1 + 2 + 2) = h(1) + h(2) + h(2)$.

There has been work on proving that by relaxing the linearity condition, a standard family of hash functions (confusingly named ``linear hashing'') behaves almost 3-wise independent~\cite{Knudsen16}. Unfortunately, the results of this paper are not strong enough for our purposes. And even if we had perfect 3-wise independence, we need higher degrees of independence to make the short cycle removal work for $k$-cycles where $k > 4$---a crucial ingredient to the hardness of distance oracles.

We propose the following solution: Instead of working over the integers, we instead work over the group~$G = \Field_p^d$ for some constant (or slightly super-constant) $p$. All the tools from additive combinatorics mentioned before work just as well over $\Field_p^d$, and also from the perspective of fine-grained complexity, the 3-SUM problem over the integers reduces to the 3-SUM problem over $\Field_p^d$:

\begin{lemma}[Integer 3-SUM to Vector 3-SUM, {{{\cite{AbboudLW14}}}}]
For any $\epsilon > 0$, there is some prime $p$ such there is no $\Order(n^{2-\epsilon})$-time algorithm for 3-SUM over $\Field_p^d$ (with $d = \Order(\log n)$), unless the 3-SUM conjecture fails.
\end{lemma}

Working over finite field vector spaces $\Field_p^d$ has the advantage that we have access to a nicer family of hash functions: Projections to random subspaces via random linear maps~\makebox{$h : \Field_p^d \to \Field_p^{d'}$}. For this family of hash functions, we can easily characterize the degree of independence: The hash values $h(a_1), h(a_2), h(a_3)$ are independent if and only if $a_1, a_2, a_3$ are linearly independent vectors. Of course the same counterexamples as above still apply, however, the number of bad triples $a_1, a_2, a_3 \in A$ is now very small: For each $a_1 \in A$, there are only $p = \Order(1)$ vectors $a_2$ which are linearly dependent on $a_1$, and similarly there are only~$p^2 = \Order(1)$ vectors $a_3$ which are linearly dependent on $a_1, a_2$. This kind of reasoning is a recurring theme in several of our proofs (see \cref{lem:energy-reduction-subsampling,lem:3sum-to-triangle}).
\section{Background on Additive Combinatorics} \label{sec:additive-combinatorics}
Additive combinatorics is the theory of additive structure in sets. In this section we summarize the basics from additive combinatorics which are needed throughout the paper. For a more thorough treatment, we refer to the book by Tao and Vu~\cite{TaoV06}. Some of the results stated here are new, because---even though the existential results are well-known---there has been no work on turning the results into efficient algorithms, to the best of our knowledge.

\subsection{Sumsets}
Recall that the sumset $A + B$ is defined as $\set{a + b : a \in A, b \in B}$. We write $r_{A, B}(x) = \#\set{(a, b) \in A \times B : a + b = x}$ to denote the \emph{multiplicities} in the sumset. As the basic building block for several upcoming proofs, we use that computing sumsets can be implemented in input- plus output-sensitive time. This fact is well-known for the integer case~\cite{ArnoldR15,ChanL15,Nakos20,BringmannFN21}, even in terms of deterministic algorithms~\cite{ChanL15,BringmannFN22}. However, we also need efficient algorithms for computing sumsets over $G = \Field_p^d$ (where for us $p = \Order(1)$ and $d = \Order(\log n)$), and to the best of our knowledge no results are known about this problem. We present the following two results, both of which follow a similar recipe than the known algorithms for integers. We postpone the proofs to \cref{sec:sparse-sumset}.

\begin{restatable}[Sparse Sumset]{lemma}{lemsparsesumset} \label{lem:sparse-sumset}
Let $G = \Field_p^d$. Given two sets $A, B \subseteq G$, we can compute $A + B$ in time \raisebox{0pt}[0pt][0pt]{$\widetilde\Order(|A + B| \cdot \poly(p, d))$} by a randomized algorithm. Moreover, the algorithm reports $r_{A, B}(x)$ for all $x \in A + B$.
\end{restatable}

\begin{restatable}[Sparse Witness Finding]{lemma}{lemsparsesumsetwitness} \label{lem:sparse-sumset-witness}
Let $G = \Field_p^d$. Given two sets $A, B \subseteq G$ and a parameter~$t$, there is a randomized algorithm running in time \raisebox{0pt}[0pt][0pt]{$\widetilde\Order(t \cdot |A + B| \cdot \poly(p, d))$} that computes, for each $x \in A + B$, a set of $t$ witnesses from $\set{(a, b) \in A \times B : a + b = x}$ (or all witnesses, if there happen to be less than $t$ many).
\end{restatable}

\subsection{Additive Energy}
An important definition for us is the \emph{additive energy $E(A)$}, defined as the number of solutions~$(a_1, a_2, a_3, a_4) \in A^4$ to the equation $a_1 + a_2 = a_3 + a_4$. Occasionally we also consider the two-set variant $E(A, B)$, defined as the number of solutions $(a_1, a_2, b_1, b_2) \in A^2 \times B^2$ to the equation $a_1 + b_1 = a_2 + b_2$. We start with some basic properties about additive energy, all of which can be proved by elementary means, see~\cite[Chapter 2]{TaoV06}.

\begin{lemma}[Additive Energy, Basic Properties] \label{lem:energy-basics}
Let $A, B \subseteq G$. Then:
\begin{itemize}
\item $|A| \cdot |B| \leq E(A, B) \leq \min(|A|^2 \cdot |B|, |A| \cdot |B|^2)$.
\item $E(A, B) = \sum_{x \in G} r_{A,B}(x)^2$.
\item $E(A, B) \geq \frac{|A|^2 |B|^2}{|A + B|}$.
\item $E(A, B) \leq E(A)^{1/2} E(B)^{1/2}$.
\end{itemize}
\end{lemma}

A property of additive energy that is particularly useful for us, is that it offers some control over the number of solutions to any linear equation (not only equations of the form~$a_1 + a_2 = a_3 + a_4$):

\begin{lemma}[Small Energy Implies Few Solutions to Linear Equations] \label{lem:energy-linear-equations}
Let $A \subseteq \Field_p^d$, let~$k \geq 4$ and fix $\alpha_1, \dots, \alpha_k, \beta \in \Field_p$ where $\alpha_1, \dots, \alpha_k \neq 0$. Then:
\begin{equation*}
    \#\set{(a_1, \dots, a_k) \in A^k : \alpha_1 a_1 + \dots + \alpha_k a_k = \beta} \leq E(A) \cdot |A|^{k-4}.
\end{equation*}
\end{lemma}
\begin{proof}
The proof is by induction on $k$. For any case $k > 4$, we can arbitrarily set the variable~$a_k$ to $|A|$ possible values and inductively bound the number of solutions to the remaining equation involving $k-1$ variables by $E(A) \cdot |A|^{k-5}$. It remains to prove the statement for the base case $k = 4$:
\begin{gather*}
    \#\set{(a_1, a_2, a_3, a_4) \in A^4 : \alpha_1 a_1 + \alpha_2 a_2 + \alpha_3 a_3 + \alpha_4 a_4 = \beta} \\
    \qquad= \#\set{(a_1, a_2, a_3, a_4) \in A^4 : \alpha_1 a_1 + \alpha_2 a_2 = \beta - \alpha_3 a_3 - \alpha_4 a_4} \\
    \qquad= \sum_{x \in \Field_p^d} r_{\alpha_1 A, \alpha_2 A}(x) \cdot r_{\beta - \alpha_3 A, -\alpha_4 A}(x)
\intertext{We apply the Cauchy-Schwartz inequality, use the identity $E(A, B) = \sum_{x \in G} r_{A, B}(x)^2$ twice, and use the bound $E(A, B) \leq E(A)^{1/2} E(B)^{1/2}$ twice:}
    \qquad= \left(\sum_{x \in \Field_p^d} r_{\alpha_1 A, \alpha_2 A}(x)^2\right)^{1/2} \cdot \left(\sum_{x \in \Field_p^d} r_{\beta - \alpha_3 A, -\alpha_4 A}(x)^2 \right)^{1/2} \\
    \qquad\leq E(\alpha_1 A, \alpha_2 A)^{1/2} \cdot E(\beta - \alpha_3 A, -\alpha_4 A)^{1/2} \\
    \qquad\leq E(\alpha_1 A)^{1/4} E(\alpha_2 A)^{1/4} E(\beta - \alpha_3 A)^{1/4} E(-\alpha_4 A)^{1/4} \\
    \qquad= E(A).
\end{gather*}
In the final step we have used that the additive energy of a set is invariant under translations and under invertible dilations.
\end{proof}

Finally, recall that from a computational perspective we often have the need to compute the additive energy. Using the identity $E(A) = \sum_{x \in G} r_{A, A}(x)^2$, and using the efficient algorithm to compute $r_{A, A}(x)$ in \cref{lem:sparse-sumset} we can compute $E(A)$ in time $\widetilde\Order(|A + A|)$. However, for unstructured sets this becomes quadratic in the size of $A$ which is prohibitive in most cases. Therefore, we typically settle for the following approximation algorithm for~$E(A)$.

\begin{lemma}[Approximating Additive Energy] \label{lem:energy-approximation}
Let $A \subseteq G$. For any constant $\epsilon > 0$, we can compute a $(1+\epsilon)$\=/approximation of $E(A)$ in time $\widetilde\Order(|A|)$ by a randomized algorithm.
\end{lemma}
\begin{proof}
Sample $R = 100 \epsilon^{-2} |A| \log |A|$ triples $a_1, a_2, a_3 \in A$, and test for each triple whether $a_1 + a_2 - a_3 \in A$. Return as an estimate $|A|^3 / R$ times the number of successful tests.

For the analysis, let $X_i$ indicate whether the $i$-th test was successful, and let $X = \sum_{i=1}^R X_i$. We have that $\Pr(X_i = 1) = E(A) / |A|^3$, and thus $\Ex(X) = R \cdot E(A) / |A|^3$. That is, our estimator is indeed unbiased.
To prove that it returns an accurate estimate with high probability, we apply Chernoff's bound:
\begin{gather*}
    \Pr\left(\left|X \cdot \tfrac{|A|^3}R - E(A)\right| \geq \epsilon E(A)\right) \\
    \qquad=\Pr\left(\left|X - \tfrac{R \cdot E(A)}{|A|^3}\right| \geq \epsilon \tfrac{R \cdot E(A)}{|A|^3}\right) \\
    \qquad\leq 2 \exp\left(-\epsilon^2 \tfrac{R \cdot E(A)}{3|A|^3}\right) \\
    \qquad\leq 2 \exp\left(-\epsilon^2 \tfrac{100\epsilon^{-2} |A| \log |A| \cdot |A|^2}{3|A|^3}\right) \\
    \qquad\leq 2 \exp(-30 \log |A|) \\
    \qquad\leq |A|^{-10}. \qedhere
\end{gather*}
\end{proof}

Most of the time we will apply \cref{lem:energy-approximation} and pretend that the output is perfect without paying too much attention to the approximation error. In all occurrences in this paper, one can easily replace the bound by, say, a $1.1$-approximation and still get the correct algorithms.

\subsection{Famous Results}
In this subsection we summarize two important results in additive combinatorics which are crucial ingredients to our algorithms. First, we recall the BSG theorem:

\thmbsg*

For an existential proof see for instance~\cite[Theorem 5]{Balog07}. 
An efficient algorithm was later devised by Chan and Lewenstein~\cite{ChanL15}, however, they designed their algorithm for a two-set version of the theorem. In \cref{sec:bsg} we detail how to conclude our version from theirs.

We will also often use the following bound due to Plünnecke~\cite{Pluennecke70} and Ruzsa~\cite{Ruzsa99} to control the size of iterated sum- and difference sets (see also \cite[Corollary 6.29]{TaoV06}).

\begin{lemma}[Plünnecke-Ruzsa Inequality] \label{lem:pluennecke}
Let $A,B \subseteq G$. If $|A + B| \leq K |A|$, then $|nB - mB| \leq K^{n+m} |A|$ for all nonnegative integers $n, m$.
\end{lemma}

Most of the time we will apply this lemma with $A = B$ (in which case the inequality is more commonly known as just \emph{Plünnecke's inequality}.)

Next, we present Ruzsa's covering lemma~\cite{Ruzsa99} and the relevant consequence that sets with small doubling can be covered by small approximate groups. We provide proofs because---even though the existential results are well-known---there has been no work on turning the results into efficient algorithms, to the best of our knowledge.

\begin{lemma}[Ruzsa's Covering Lemma]
Let $A, B \subseteq G$. Then there is a subset $X \subseteq B$ with the following properties:
\begin{itemize}
\item $B \subseteq A - A + X$,
\item $|X| \leq \frac{|A + B|}{|A|}$.
\end{itemize}
Moreover, we can compute $X$ in time $\widetilde\Order(\frac{|A - A + B| \cdot |A + B|}{|A|}) \leq \widetilde\Order(\frac{|A - A + B|^2}{|A|})$. 
\end{lemma}
\begin{proof}
We start with a recap of the well-known existential proof. The proof is in fact already algorithmic: We initialize the set $X \gets \emptyset$. While there exists some $b \in B$ such that $A + b$ is disjoint from $A + X$, add $b$ to~$X$.

We prove that after the algorithm has terminated, $X$ is as desired. Indeed, after the algorithm has terminated, we have that $A + b$ and $A + X$ are not disjoint for any $b \in B$. Or equivalently, $b \in A - A + X$. Moreover, note that the size of $A + X$ increases by $|A|$ with every step of the algorithm and that ultimately~$|A + X| \leq |A + B|$. It follows that the algorithm runs for at most \raisebox{0pt}[0pt][0pt]{$\frac{|A + B|}{|A|}$} iterations. Since each iteration adds exactly one element to $X$, we obtain the claimed size bound \raisebox{0pt}[0pt][0pt]{$|X| \leq \frac{|A + B|}{|A|}$}.

While this proof is already algorithmic, it is a priori not clear how to efficiently find~$b$. Our approach is as follows: Compute the sets $C \gets (A + B) \setminus (A + X)$ and $C - A$, and additionally compute the multiplicities~$r_{C, -A}(x)$ for all $x \in C - A$. We now take any~\makebox{$b \in B$} satisfying~$r_{C, -A}(b) = |A|$, and if no such $b$ exists we terminate the algorithm. Recall that~$r_{C, -A}(b)$ is equal to the number of witnesses~$(c, a) \in C \times A$ with~$a + b = c$. There are $|A|$ such witnesses (the maximum number) if and only if $A + b \subseteq C$. By the way we assigned $C$, this in turn is equivalent to the desired condition that $A + b$ is disjoint from $A + X$.

It remains to analyze the running time of this algorithm. Finding a single $b$ amounts to computing the sets $C \subseteq A + B$ and $C - A \subseteq A - A + B$. Using \cref{lem:sparse-sumset} we can compute both sets in output-sensitive time~$\widetilde\Order(|A - A + B|)$, and compute the multiplicities $r_{C, -A}$ in the same time. Finally, recall that the algorithm runs for a total of \raisebox{0pt}[0pt][0pt]{$\frac{|A + B|}{|A|}$} iterations. The claimed time bound follows.
\end{proof}

\lemcoverapproximategroup*
\begin{proof}
We first apply Ruzsa's covering lemma with $A$ and $B = 2A - 2A$. We thereby obtain a subset~\makebox{$X \subseteq 2A - 2A$} which satisfies that $B \subseteq A - A + X$. By choosing $H = A - A$, we have that $H + H = B \subseteq A - A + X = H + X$. Ruzsa' covering lemma further guarantees that
\begin{equation*}
    |X| \leq \frac{|A + B|}{|A|} = \frac{|3A - 2A|}{|A|} \leq \frac{K^5|A|}{|A|} = K^5,
\end{equation*}
where for the latter inequality we have applied Plünnecke's inequality. Therefore, $H$ satisfies the second property. The first property is easy by another application of Plünnecke's inequality. For the third take an arbitrary $a_0 \in A$. Then by definition $A - a_0 \subseteq A - A = H$.

The running time to compute $H$ and $X$ is dominated by the call to Ruzsa's covering lemma, which runs in time
\begin{equation*}
    \Order\left(\frac{|A - A + B|^2}{|A|}\right) = \Order\left(\frac{|3A - 3A|^2}{|A|}\right) \leq \Order(K^{12} |A|),
\end{equation*}
where have again used Plünnecke's inequality.
\end{proof}

\subsection{Linear Hashing}
Another important tool that we frequently use throughout (though not related to additive combinatorics) is \emph{linear hashing}. We say that a hash function $h : G \to G'$ is linear if it satisfies $h(a + b) = h(a) + h(b)$. Since we are most interested in the case $G = \Field_p^d$ and~$G' = \Field_p^{d'}$ for some $d' \ll d$, we often make use of the following simple construction:

\begin{lemma}[Hashing via Random Linear Maps] \label{lem:linear-hashing}
Let $h : G \to G'$ be a random linear map (i.e., let $H \in \Field_p^{d' \times d}$ be a random matrix, and let $h(x) = Hx$). Then the following properties are satisfied:
\begin{itemize}
\item \emph{Linearity:} $h(a + b) = h(a) + h(b)$ for all $a, b \in \Field_p^d$.
\item \emph{Independence:} For any linearly independent vectors $a_1, \dots, a_k \in \Field_p^d$ (in particular, the $a_i$'s must be nonzero), the random variables $h(a_1), \dots, h(a_k)$ are independent, and for any~\smash{$x_1, \dots, x_k \in \Field_p^{d'}$} we have 
\begin{equation*}
    \Pr(\text{$h(a_1) = x_1$ and $\dots$ and $h(a_k) = x_k$}) = (p^{d'})^{-k},
\end{equation*}
More generally, for any $a_1, \dots, a_k \in \Field_p^d$ (not necessarily linearly independent) and any~$x_1, \dots, x_k \in \Field_p^{d'}$ we have that
\begin{equation*}
    \Pr(\text{$h(a_1) = x_1$ and $\dots$ and $h(a_k) = x_k$}) \leq (p^{d'})^{-s},
\end{equation*}
where $s = \dim\Span{a_1, \dots, a_k}$.
\end{itemize}
\end{lemma}
\begin{proof}
The first property is obvious. To prove the second statement, first recall that any set of linearly independent vectors $a_1, \dots, a_k$ can be written as $a_i = M e_i$, where $M$ is a full-rank matrix and $e_i$ is the all-zeros vector with a single $1$ in position $i$. Next, observe that the matrix $H M$ is uniformly random (indeed for any fixed matrix $N$ we have that $\Pr(HM = N) = \Pr(H = N M^{-1})$ and $H$ is uniformly random). It follows that the hash values~\makebox{$h(a_i) = H M e_i$} are the columns of a uniformly random matrix and therefore independent. Hence:
\begin{equation*}
    \Pr(\text{$h(a_1) = x_1$ and $\dots$ and $h(a_k) = x_k$}) = (p^{d'})^{-k},
\end{equation*}
for any $x_1, \dots, x_k \in \Field_p^{d'}$.

To obtain the more general statement for vectors which are not necessarily linearly independent, select a subset from $\set{a_1, \dots, a_k}$ of $\dim\Span{a_1, \dots, a_k}$ linearly independent vectors. For this subset, the hash values behave independently.
\end{proof}
\section{3-SUM for Structured Inputs} \label{sec:3sum-structured}

The purpose of this section is to prove the following theorem. We focus on the group $G = \Field_p^d$, but the theorem holds for $G = \Int$ as well, using the same proof with minor modifications (such as substituting an appropriate linear hash function for integers). 

\thmthreesumstructured*

\paragraph{Universe Reduction}
As the first step, we will hash all sets to a smaller group $G'$ of size $\ll n^2$ via some hash function $h : G \to G'$. Under the hashing we are bound to introduce several \emph{false positives}, that is, triples $a \in A, b \in B, c \in C$ where $a + b + c \neq 0$ but $h(a) + h(b) + h(c) = 0$. To deal with these false positives, we have to list several 3-SUM solutions in the smaller group instead of merely determining the existence of one solution. The following problem definition and lemma make this precise.

\begin{definition}[3-SUM Listing]
Given sets $A, B, C \subseteq G$ and a parameter $t$, compute for each~$a \in A$ a list of $t$ distinct pairs $(b, c) \in B \times C$ with $a + b + c = 0$ (or if there are less than $t$ solutions, a list containing all of them).
\end{definition}

\begin{lemma}[Reduction to 3-SUM Listing in Small Groups] \label{lem:3sum-universe-reduction}
Let $G = \Field_p^d$ and let $G' = \Field_p^{d'}$, and let~\makebox{$h : G \to G'$} be a random linear map. There is a fine-grained reduction from a 3-SUM instance $A, B, C \subseteq G$ to one 3-SUM listing instance $h(A) = \set{h(a) : a \in A}, h(B), h(C) \subseteq G'$ with parameter $t = \Order(n^2 / |G'|)$. The reduction runs in time $\Order(n t^2)$ and succeeds with constant probability $\frac{8}{10}$.
\end{lemma}
\begin{proof}
Let $t = 10 n^2 |G'|^{-1}$. We may assume that $|G'| \geq 10 n$, since otherwise we can simply solve the given instance in time $\Order(n^2) = \Order(n t)$. We precompute a lookup table to find, given a hash value $a' \in h(A)$, all~$a \in A$ with $h(a) = x$ (and similarly for $B$ and~$C$). Then we run the listing algorithm with parameter $t$ on the instance $(h(A), h(B), h(C))$, and for every reported solution~$a', b', c'$ we use the lookup table to check whether these correspond to some $a \in A, b \in B, c \in C$ with $a + b + c = 0$.

It is clear that the reduction cannot report ``yes'' unless the given 3-SUM instance $(A, B, C)$ is a ``yes'' instance. We argue that the reduction misses a ``yes'' instance with probability at most~$\frac{1}{10}$. To this end we fix any $a \in A$ which is part of a 3-SUM solution and prove that with probability at least~$\frac{9}{10}$, there are less than~$t$ many false positives~$(b^*, c^*) \in B \times C$ with $h(a) + h(b^*) + h(c^*) = 0$. In this case it follows that the listing algorithm will return at least one proper solution $(h(a), h(b), h(c))$ where $a + b + c = 0$, and we will recover $(a, b, c)$ using the lookup table. And indeed, for any fixed pair $(b^*, c^*) \in B \times C$ with $a + b^* + c^* \neq 0$, we have that $h(a) + h(b^*) + h(c^*) = 0$ with probability at most $|G'|^{-1}$. Hence the expected number of false positives is $|B| \cdot |C| \cdot |G'|^{-1} \leq n^2 |G'|^{-1}$. By Markov's inequality the number of false positives exceeds~$t = 10n^2 |G'|^{-1}$ with probability at most $\frac{1}{10}$. This completes the correctness argument.

Before analyzing the running time, we first analyze the maximum \emph{bucket load $L$} of the hashing, i.e., the maximum number of elements in $A$ (or similarly in $B$ or $C$) hashing to the same value under $h$. Note that~$\binom L2$ is at most the number of \emph{collisions} of the hash function (i.e., the number of distinct pairs $a, a' \in A$ with~$h(a) = h(a')$), as any two elements in the same bucket cause a collision. For any fixed $a, a' \in A$, the collision probability is at most~$|G'|^{-1}$ and thus the expected number of collisions is at most $n^2 |G'|^{-1}$. Using again Markov's inequality, the hashing causes at most $t = 10 n^2 |G'|^{-1}$ collisions with probability~$\frac{9}{10}$, and in that case we can bound $L = \Order(t^{1/2})$.

We are finally ready to bound the running time. Constructing $h(A), h(B), h(C)$, the hashing and the lookup table takes linear time. After that, we check all the $nt$ listed solutions. For each solution $(x, y, z)$ we have to enumerate all pairs $(a, b) \in A \times B$ with~\makebox{$h(a) = x$} and $h(b) = y$ which takes time $\Order(L^2)$. Hence, the total time is $\Order(n t^2)$.
\end{proof}

\paragraph{The Algorithm}
Using \cref{lem:3sum-universe-reduction}, we may assume that $A, B, C \subseteq G$ where $G$ has size~$\Order(n^2 / t)$, and we have to list~$t$ solutions for each $a \in A$. Moreover, by the assumption in \cref{thm:3sum-structured} we can assume that~$|A + A| \leq K|A|$ (this property is preserved under the linear hashing in \cref{lem:3sum-universe-reduction}).

The 3-SUM algorithm is given in \cref{alg:3sum-structured}. First cover $A$ by a translate of an approximate group~$H$ (that is, a small set $H$ satisfying $H + H \subseteq H + X$ where $X$ is small) using \cref{lem:cover-approximate-group}. Then sample a set $S \subseteq G$ (with rate \raisebox{0pt}[0pt][0pt]{$\frac{10 \log n}{|H|}$}) such that $H + S$ covers the whole universe $G$. We precompute the sets $B_s = B \cap (H + s)$ and $C_{s, x} = C \cap (H - s - x)$ for all shifts~$s \in S$ and~$x \in X$. The crucial insight is that we only have to look for 3-SUM solutions in $A \times B_s \times C_{s, x}$. For each such group we apply a heavy-light approach: Either the sets $B_s, C_{s, x}$ are sparse (with size smaller than some parameter $\Delta$ to be determined later), and we can afford to enumerate all pairs. Or the sets are dense, in which case we compute $B_s + C_{s, x}$ in linear time, but this case cannot happen too often. We analyze these steps in more detail, starting with the proof that $H + S$ indeed covers the whole universe $G$:

\begin{algorithm}[t]
\caption{Lists $t$ 3-SUM solutions for a given instance $A, B, C \subseteq G$ where $|G| \leq \Order(n^2 / t)$ and $|A + A| \leq K |A|$ in subquadratic time.} \label{alg:3sum-structured}
\begin{algorithmic}[1]
\State Apply \cref{lem:cover-approximate-group} on $A$ to compute $H$, $X$ and $a_0$ \label{alg:3sum-structured:line:cover}
\State Subsample a set $S \subseteq G$ with rate $\frac{100\log n}{|H|}$ \label{alg:3sum-structured:line:shifts}
\ForEach{$s \in S, x \in X$}
    \State Compute the sets $B_s = B \cap (H + s)$ and $C_{s, x} = C \cap (H - s - x - a_0)$ \label{alg:3sum-structured:line:precompute}
\EndForEach
\ForEach{$s \in S, x \in X$} \label{alg:3sum-structured:line:loop}
    \If{$|B_s| \leq \Delta$ and $|C_{s, x}| \leq \Delta$}
        \State List all pairs $(b, c) \in B_s \times C_{s, x}$ and whenever $b + c \in -A$, report the corresponding
        \Statex[2] 3-SUM solution $(-(b+c), b, c)$ \label{alg:3sum-structured:line:enumerate}
    \Else
        \State Compute the sumset $B_s + C_{s, x}$ using \cref{lem:sparse-sumset,lem:sparse-sumset-witness} and list $t$ witnesses $(b, c)$
        \Statex[2] for each $a \in B_s + C_{s, x}$ \label{alg:3sum-structured:line:witness}
        \ForEach{$a \in (B_s + C_{s, x}) \cap -A$}
            \State Report a 3-SUM solution $(-a, b, c)$ for each witness $(b, c)$ of $a$ \label{alg:3sum-structured:line:witness-report}
        \EndForEach
    \EndIf
\EndForEach
\end{algorithmic}
\end{algorithm}

\begin{lemma}[Random Cover] \label{lem:random-cover}
With high probability, for any $z \in G$ there are~$\Theta(\log n)$ shifts~$s \in S$ such that $z \in H + s$ (in short: $r_{H, S}(z) = \Theta(\log n)$).
\end{lemma}
\begin{proof}
In expectation, each element $z \in G$ is contained in $|H| \cdot \frac{100 \log n}{|H|} = 100 \log n$ sets of the form $H + s, s \in S$. By Chernoff's bound, the probability that we hit less than $50 \log n$ sets or more than $150 \log n$ sets is at most $2\exp(-\frac{100 \log n}{12}) \leq n^{-8}$. Taking a union bound over the $|G| \leq n^2$ elements $z$, the statement is correct with probability at least $1 - n^{-6}$.
\end{proof}

\begin{lemma}[Correctness of \cref{alg:3sum-structured}] \label{lem:3sum-structured-correctness}
\cref{alg:3sum-structured} is correct, that is, it reports a list of $t$ witnesses for each $a \in A$ (or a list of all witnesses if there are less than $t$ many).
\end{lemma}
\begin{proof}
Focus on any 3-SUM solution $a + b + c = 0$. The key is to prove that there are shifts~$s \in S, x \in X$ such that $(a, b, c) \in A \times B_s \times C_{s, x}$. In this case it is easy to check that the algorithm will either report $(a, b, c)$ (in \cref{alg:3sum-structured:line:enumerate} or \cref{alg:3sum-structured:line:witness-report}) or it already reported $t$ other witnesses for~$a$ (if the list of $t$ witnesses computed in \cref{alg:3sum-structured:line:witness} does not contain the particular witness $(b, c)$).

To see that $s, x$ exist as claimed, invoke the previous lemma to find some $s \in S$ such that~$b \in H + s$, that is, there is some $v \in H$ such that~$b = v + s$. Then, since $a \in A \subseteq H + a_0$ by \cref{lem:cover-approximate-group} we have that $a + v \in H + H + a_0 \subseteq H + X + a_0$. Thus, there is some $w \in H$ and $x \in X$ such that~$a + v = w + x + a_0$. It follows that $c = -(a + b) = -(w + s + x + a_0) \in -H - s - x - a_0 = H - s - x - a_0$. Using the definitions of $B_s$ and $C_{s, x}$, we conclude that~$b \in B_s$ and $c \in C_{s, x}$ as stated. 
\end{proof}

\begin{lemma}[Running Time of \cref{alg:3sum-structured}] \label{lem:3sum-structured-time}
\cref{alg:3sum-structured} runs in time
\begin{equation*}
    \widetilde\Order\left(K^{12} \left(\frac{n^2}{t} + \frac{n\Delta^2}{t} + \frac{n^2 t}{\Delta} \right)\right).
\end{equation*}
\end{lemma}
\begin{proof}
Computing $H$ and $X$ in \cref{alg:3sum-structured:line:cover} takes time $\Order(K^{12} n)$ by \cref{lem:cover-approximate-group}, and the same lemma guarantees that $|H| \leq K^2 n$ and $|X| \leq K^5$. Sampling $S$ in a naive way in \cref{alg:3sum-structured:line:shifts} takes time $\Order(|G|) = \Order(n^2 / t)$ and with high probability, $S$ has size \raisebox{0pt}[0pt][0pt]{$\widetilde\Order(\frac{|G|}{|H|}) = \widetilde\Order(n/t)$}. In \cref{alg:3sum-structured:line:precompute}, it takes linear time to compute each set $B_s$ and~$C_{s, x}$, so the total time is $\Order(n |S| |X|) \leq \widetilde\Order(K^5 n^2 / t)$.

For the loop over pairs $s \in S, x \in X$ (in \cref{alg:3sum-structured:line:loop}) we split the analysis into two cases: The \emph{light} pairs $s, x$ with $|B_s|, |C_{s, x}| \leq \Delta$ and the remaining \emph{heavy} pairs. There are up to~$|S| \cdot |X| \leq \widetilde\Order(K^5 n / t)$ light pairs, and for each such pair we spend time $\Order(\Delta^2)$ in \cref{alg:3sum-structured:line:enumerate}. The total time spent on light pairs is thus $\Order(K^5 n \Delta^2 / t)$.

The number of heavy pairs is bounded by $\widetilde\Order(K^5 n / \Delta)$. Indeed, recall that each element~$b$ occurs in at most $\Order(\log n)$ sets $B_s$ by \cref{lem:random-cover}. Hence, there are at most $\widetilde\Order(n / \Delta)$ heavy sets $B_s$. Similarly, each element~$c$ occurs in at most $\Order(|X| \log n) = \Order(K^5 \log n)$ sets $C_{s, x}$ and therefore the number of heavy sets $C_{s, x}$ is at most $\widetilde\Order(K^5 n / \Delta)$. For each heavy pair, we spend time \raisebox{0pt}[0pt][0pt]{$\widetilde\Order(t \cdot |B_s + C_{s, x}|)$} to list $t$ witnesses for each element in the sumset $B_s + C_{s, x}$ by \cref{lem:sparse-sumset-witness}. Recall that $B_s + C_{s, x} \subseteq H + H - x - a_0 \subseteq H + X - x - a_0$, and thus $|B_s + C_{s, x}| \leq |H| \cdot |X| \leq K^7 n$. Hence, the heavy pairs amount to time $\Order((K^5 n / \Delta) \cdot (t K^7 n)) = \Order(K^{12} n^2 t / \Delta)$. Summing over all these contributions gives the claimed time bound.
\end{proof}

\begin{proof}[Proof of \cref{thm:3sum-structured}]
We proceed as outlined before: First apply \cref{lem:3sum-universe-reduction} to reduce the given 3-SUM instance to a 3-SUM listing instance with parameter $t$ in a universe of size~$|G'| = \Order(n^2 / t)$ and then run \cref{alg:3sum-structured} on that instance. This algorithm is correct by \cref{lem:3sum-structured-correctness}, and runs in the claimed running time by setting $t = n^{1/4}$ and $\Delta = n^{1/2}$. Since the universe reduction succeeds only with constant probability~$\frac{8}{10}$, we need to repeat this whole process $\Order(\log n)$ times to achieve high success probability.
\end{proof}
\section{Energy Reduction for 3-SUM} \label{sec:3sum-energy-reduction}
We prove the energy reduction in two steps, as outlined in the overview.
\lemenergyreductionadditivecombinatorics*
\begin{proof}
The reduction is given in \cref{alg:energy-reduction-additive-combinatorics}. We repeatedly estimate the additive energy of~$A$ using \cref{lem:energy-approximation} and as long as~$E(A) \geq |A|^3 / K$, we apply the BSG theorem to obtain a structured subset~$A' \subseteq A$. This set has large size $|A'| \geq \Omega(|A| / K)$ and small doubling~\makebox{$|A' + A'| \leq \Order(K^{24} |A|) = \Order(K^{26} |A'|)$}. We solve the 3-SUM instance~$(A', A, A)$ using \cref{thm:3sum-structured}; if a solution is found in this step we report ``yes''. Otherwise continue the process with $A \setminus A'$ in place of $A$. As soon as the additive energy of $A$ drops below the desired threshold $|A|^3 / K$, we stop and return $A^* \gets A$.

The correctness is easy to prove: In each step we split off a subset $A'$. If there is a 3-SUM solution involving an element from $A'$, we detect the solution by calling \cref{thm:3sum-structured} and correctly report ``yes''. Otherwise it is safe to discard $A'$. 

To analyze the running time, first observe that in every step the size of $A$ reduces by at least $\Omega(K^{-2} |A|)$. Therefore, after at most $\Order(K^2)$ steps the size of $A$ must have halved and thus the total number of steps is bounded by $\Order(K^2 \log n)$. In each step, computing $A'$ via the BSG theorem takes time $\Order(K^{12} |A|)$ and solving the structured 3-SUM instance $(A', A, A)$ takes time \raisebox{0pt}[0pt][0pt]{$\widetilde\Order((K^{26})^{12} n^{7/4}) = \widetilde\Order(K^{312} n^{7/4})$}. In total we spend time \raisebox{0pt}[0pt][0pt]{$\widetilde\Order(K^{314} n^{7/4})$} as claimed.
\end{proof}

\begin{algorithm}[t]
\caption{The energy reduction via additive combinatorics. Given a 3-SUM instance $A \subseteq G$, this algorithm either detects a 3-SUM solution or constructs an equivalent instance $A^* \subseteq A$ with additive energy $E(A^*) \leq |A^*|^3 / K$.} \label{alg:energy-reduction-additive-combinatorics}
\begin{algorithmic}[1]
\RepeatInf
    \State Estimate $E(A)$ using \cref{lem:energy-approximation}
    \If{$E(A) \leq |A|^3 / K$}
        \State\Return $A^* \gets A$
    \Else
        \State Apply the Balog-Szemerédi-Gowers theorem on $A$ to obtain $A' \subseteq A$
        \State Solve the 3-SUM instance $(A', A, A)$ using \cref{thm:3sum-structured}
        \If{$(A', A, A)$ is a ``yes'' instance} \Return ``yes'' \EndIf
        \State $A \gets A \setminus A'$
    \EndIf
\EndRepeatInf{}
\end{algorithmic}
\end{algorithm}

\lemenergyreductionsubsampling*

\begin{algorithm}[t]
\caption{The energy reduction via subsampling. Given a 3-SUM instance $A \subseteq G = \Field_p^d$ with bounded additive energy $E(A) \leq |A|^3 / K$, this algorithm constructs \raisebox{0pt}[0pt][0pt]{$\Order(|A|^2 / K^2)$} smaller 3-SUM instances of size $\Order(K)$ and with expected additive energy $\Order(K^2)$.} \label{alg:energy-reduction-subsampling}
\begin{algorithmic}[1]
\State Let $d' = \ceil{\log_p(|A|/K)}$ and let $G' = \Field_p^{d'}$
\State Sample a linear hash function $h: G \to G'$
\ForEach{$x, y \in G'$}
    \State Construct the 3-SUM instance $A_{x, y} = \set{a \in A : h(a) \in \set{x, y, -(x+y)}}$ \label{alg:energy-reduction-subsampling:line:construct}
    \If{$|A_{x, y}| \leq 6K$} \label{alg:energy-reduction-subsampling:line:exceeds-condition}
        \State Solve the 3-SUM instance $A_{x, y}$ by means of the reduction \label{alg:energy-reduction-subsampling:line:reduction}
    \Else
        \State Solve the 3-SUM instance $A_{x, y}$ by brute-force \label{alg:energy-reduction-subsampling:line:brute-force}
    \EndIf
\EndForEach
\State\Return ``yes'' if and only if one of the instances $A_{x, y}$ is a ``yes'' instance
\end{algorithmic}
\end{algorithm}

\noindent
The reduction is summarized in \cref{alg:energy-reduction-subsampling}. We sample a linear hash function $h : G \to G'$ and construct the instances $A_{x, y} = \set{a \in A : h(a) \in \set{x, y, -(x + y)}}$, for all $x, y \in G'$. We solve all instances by brute-force which exceed their expected size by a constant factor, and pass the other instance to the reduction. If we find a 3-SUM solution in one of the constructed instances, we report ``yes''.

The analysis involves several steps, but the correctness argument is simple: Since all sets~$A_{x, y}$ are subsets of $A$, we can never return ``yes'' unless $A$ is a ``yes'' instance. On the other hand, whenever there is a 3-SUM solution~$a + b + c = 0$ in $A$, we can pick $x = h(a)$ and $y = h(b)$ so that $A_{x, y}$ is a ``yes'' instance (by the linearity of the hash function).

We continue with the analysis of the running time of the reduction, which mainly involves proving that most instances have size $\Order(K)$ and therefore do not have to be brute-forced.

\begin{lemma}[Running Time of \cref{alg:energy-reduction-subsampling}]
\Cref{alg:energy-reduction-subsampling} runs in expected time $\widetilde\Order(n^2 / K)$.
\end{lemma}
\begin{proof}
For most steps of the algorithm it is easy to bound the running time. In particular, we can construct the instances $A_{x, y}$ in time $\Order(n^2 / K)$ by first precomputing the hash values~$h(a)$ for all $a \in A$. The interesting part is to bound the running time of the brute-force step in \cref{alg:energy-reduction-subsampling:line:brute-force}. To this end, we analyze the sizes of the constructed instances $A_{x, y}$.

Fix any $x, y \in G'$. We compute the expectation and variance of $|A_{x, y}|$ as follows. For ease of notation, write $X = \set{x, y, -(x + y)}$:
\begin{equation*}
    \Ex(|A_{x, y}|) = \sum_{a \in A} \Pr(h(a) \in X) \leq \sum_{a \in A} \frac{3}{|G'|} = \frac{3n}{|G'|} \leq 3K.
\end{equation*}
Next, we compute the variance:
\begin{gather*}
    \Var(|A_{x, y}|) \\
    \qquad= -\Ex(|A_{x, y}|)^2 + \Ex(|A_{x, y}|^2) \\
    \qquad= -\left(\sum_{a \in A} \Pr(h(a) \in X)\right)^2 + \sum_{a, b \in A} \Pr(h(a), h(b) \in X)
\intertext{Here, we distinguish two cases for $a, b$: If $a$ and $b$ are linearly independent, then the random variables $h(a)$ and $h(b)$ are independent. If $a, b$ are linearly dependent, then there are at most $np = \Order(n)$ choices for~$a, b$ (fix $a$ arbitrarily, then there are at most $p$ choices for $b$ in the span $\Span{a}$). It follows that the above expression can be bounded as follows:}
    \qquad\leq -\left(\sum_{a \in A} \Pr(h(a) \in X)\right)^2 + \left(\sum_{a, b \in A} \Pr(h(a) \in X) \cdot \Pr(h(b) \in X)\right) + \Order\left(\frac{n}{|G'|}\right) \\
    \qquad\leq \Order\left(\frac{n}{|G'|}\right) \\
    \qquad\leq \Order(K).
\end{gather*}

We are now ready to bound the expected running time of \cref{alg:energy-reduction-subsampling:line:brute-force} using Chebyshev's inequality:
\begin{gather*}
    \sum_{x, y \in G'} \sum_{i=0}^{\log n} \Pr(|A_{x, y}| \geq 2^i \cdot 6K) \cdot \Order((2^i K)^2) \\
    \qquad\leq \sum_{x, y \in G'} \sum_{i=0}^{\log n} \Pr\Big(|A_{x, y}| - \Ex(|A_{x, y}|) \geq \Omega(2^i \Var(|A_{x, y}|))\Big) \cdot \Order(2^{2i} K^2) \\
    \qquad\leq \sum_{x, y \in G'} \sum_{i=0}^{\log n} \Order\left(\frac{1}{2^{2i} K} \cdot 2^{2i} K^2 \right) \\
    \qquad\leq \widetilde\Order(n^2 / K).
\end{gather*}
This completes the running time analysis.
\end{proof}

\begin{lemma}[Bounded Energy] \label{lem:energy-bounded}
Fix $x, y \in G'$ and let $A_{x, y}$ be as in \cref{alg:energy-reduction-subsampling}. Then $\Ex(E(A_{x, y})) \leq \Order(K)$.
\end{lemma}
\begin{proof}
We bound the expected energy as follows:
\begin{gather*}
    \Ex(E(A_{x, y})) \\
    \qquad= \sum_{\substack{a_1, a_2, a_3, a_4 \in A\\a_1 + a_2 = a_3 + a_4}} \Pr(h(a_1), h(a_2), h(a_3), h(a_4) \in \set{x, y, -(x - y)} ) \\
    \qquad= \sum_{s=0}^3 \sum_{\substack{a_1, a_2, a_3, a_4 \in A\\a_1 + a_2 = a_3 + a_4\\\dim\Span{a_1, a_2, a_3, a_4} = s}} \Pr(h(a_1), h(a_2), h(a_3), h(a_4) \in \set{x, y, -(x - y)} )
\intertext{For fixed elements $a_1, a_2, a_3, a_4$ spanning a subspace of dimension $s$, there are at least $s$ hash values in $h(a_1), h(a_2), h(a_3), h(a_4)$ which are independent and therefore the probability can be upper bounded by $1/|G'|^s$.}
    \qquad\leq \sum_{s=0}^3 \frac{1}{|G'|^s} \cdot \sum_{\substack{a_1, a_2, a_3, a_4 \in A\\a_1 + a_2 = a_3 + a_4\\\dim\Span{a_1, a_2, a_3, a_4} = s}} 1
\intertext{We now distinguish two cases: For $s = 3$ we use that the inner sum is at most $E(A)$ by definition. For $s \leq 2$ we bound the inner sum by the weaker bound $\Order(|A|^s)$. (Indeed, any tuple~$a_1, a_2, a_3, a_4$ spanning a subspace of dimension~$s$ can be obtained by first picking~$s$ arbitrary elements from $A$ and expressing the others as one out of $p^4 = \Order(1)$ possible linear combinations.)}
    \qquad\leq \sum_{s=0}^2 \Order\left(\frac{|A|^s}{|G'|^s}\right) + \Order\left(\frac{E(A)}{|G'|^3}\right) \\
    \qquad\leq \Order\left(\frac{|A|^2}{|G'|^2} + \frac{E(A)}{|G'|^3}\right) \\
    \qquad\leq \Order(K^2).
\end{gather*}
In the final step we have used that $|G'| \geq |A| / K$ and that $E(A) \leq |A|^3 / K$.
\end{proof}

By concatenating both energy reductions we obtain the following result.

\thmenergyreduction*
\begin{proof}
Suppose that there are $\epsilon, \delta > 0$ and an algorithm $\mathcal A$ solving 3-SUM on instances~$A$ of size $n$ with additive energy $\Order(|A|^{2+\delta})$ in time $\Order(n^{2-\epsilon})$.

We reduce a given 3-SUM instance $A$ to this problem. Let $K = |A|^{0.0001}$. We first apply \cref{lem:energy-reduction-additive-combinatorics} with parameter $K$ to either detect a 3-SUM solution in $A$ or to find an equivalent instance $A^* \subseteq A$ with additive energy bounded by $|A^*|^3 / K$.

Next, apply the reduction from \cref{lem:energy-reduction-subsampling} to obtain $g = \Order(|A|^2 / K^2)$ instances~$A_1, \dots, A_g$ of size $\Order(K)$ with expected additive energy $\Order(K^2)$. By Markov's bound, each such instance has additive energy more than $K^{2+\delta}$ with probability at most $\Order(K^{-\delta})$. We may therefore use \cref{lem:energy-approximation} to estimate the additive energies of the constructed instances, and brute-force all instances with energy exceeding $K^{2+\delta}$. We solve the remaining instances using the efficient algorithm~$\mathcal A$.

It remains to analyze the running time. \cref{lem:energy-reduction-additive-combinatorics} runs in time $\Order(K^{314} n^{7/4}) = \Order(n^{1.7814})$ and \cref{lem:energy-reduction-subsampling} runs in time $\Order(n^2 / K) = \Order(n^{1.9999})$. Since we only solve a $K^{-\epsilon}$-fraction of the instances by brute-force, the total expected running time of brute-forcing instances with exceptionally large additive energy takes time $\Order(K^{-\epsilon} n^2 / K^2 \cdot K^2) = \Order(n^{2-0.0001\epsilon})$. Finally, solving the remaining instances using~$\mathcal A$ amounts for time~\makebox{$\Order(n^2 / K^2 \cdot K^{2-2\delta}) = \Order(n^{2-0.0001\delta})$}. All in all, the running time is subquadratic as claimed.
\end{proof}
\section{Reducing 3-SUM to Triangle Listing} \label{sec:3sum-to-triangle}
The first reduction from 3-SUM to triangle listing is by Pătraşcu~\cite{Patrascu10}, and this reduction was later generalized by Kopelowitz, Pettie and Porat~\cite{KopelowitzPP16}. It is also known how to adapt the reduction to 3-XOR~\cite{JafargholiV16} (i.e., the~$G = \Field_2^d$ version of 3-SUM).

In this section we revisit this reduction. We present a modified (and arguably simplified) version of the known constructions. As before, we consider 3-SUM instances over the group~$G = \Field_p^d$, where $p$ is a constant prime and~$d = \Order(\log n)$. Our goal is to prove the following theorem:

\thmthreesumtotriangle*

For the remainder of this subsection, we will prove \cref{lem:3sum-to-triangle}. We start with the construction in \cref{sec:3sum-to-triangle:sec:construction}. In \cref{sec:3sum-to-triangle:sec:cycles} we analyze the number of $k$-cycles and in \cref{sec:3sum-to-triangle:sec:regular} we justify the assumption that the graph is $\Theta(n^{1/2})$-regular. We summarize the proof of \cref{lem:3sum-to-triangle} in \cref{sec:3sum-to-triangle:sec:assembling}. Throughout, let~$A$ be the given 3-SUM instance. By the energy reduction in \cref{lem:energy-reduction-complete} (applied with $\delta = \frac12$, say) we can assume that $E(A) \leq \Order(|A|^{5/2})$.

\subsection{The Construction} \label{sec:3sum-to-triangle:sec:construction}
We start with the construction of the triangle listing instance. Let $G' = \Field_p^{d'}$ be a subspace of $G$ with prescribed size~$|G'| \leq n$ which we will set later. We randomly sample linear maps~$h_1, h_2, h_3 : G \to G'$, and let~$h : G \to (G')^3$ be defined by $h(a) = (h_1(a), h_2(a), h_3(a))$. Let
\begin{equation*}
    \begin{array}{c@{\;}c@{\;}c@{\:}c@{\:}c@{\:}c@{\:}c}
        X & = & G' & \times & G' & \times & \set{0} \\[.5ex]
        Y & = & G' & \times & \set{0} & \times & G' \\[.5ex]
        Z & = & \set{0} & \times & G' & \times & G'
    \end{array}
\end{equation*}
be the vertex parts in the constructed tripartite graph. Observe that each set $X, Y, Z$ is a subgroup of~$(G')^3$. We now add edges to the graph: For each $a \in A$, add an edge between~$x \in X$ and~$y \in Y$ whenever~\makebox{$y = x + h(a)$}. We say that this edge $(x, y)$ is \emph{labeled} with $a$. Similarly, add an edge between $y \in Y$ and $z \in Z$ whenever~$z = y + h(a)$ and add an edge between $z \in Z$ and $x \in X$ whenever~$x = z + h(a)$. We remark that for the analysis we view the instance as a labeled (multi-)graph with labels as just described, but for the actual reduction we forget about the edge labels (and multiple edges) and treat the constructed instance as a simple graph; this notation is purely for convenience.

We introduce some more notation. As before, we say that $(a, b, c) \in A^3$ is a \emph{solution} if~$a + b + c = 0$. We say that $(a, b, c) \in A^3$ is a \emph{pseudo-solution} if $h(a) + h(b) + h(c) = 0$. As a first step, we argue that there is a one-to-one correspondence between triangles in the constructed instance and pseudo-solutions.

\begin{lemma}[Pseudo-Solutions Are Triangles] \label{lem:pseudo-solutions}
The labels $a, b, c$ of any triangle in the constructed instance form a pseudo-solution. Moreover, for every pseudo-solution $a, b, c$ there are at most six triangles in the instance labeled with $a, b, c$.
\end{lemma}
\begin{proof}
The first claim is easy: By the construction of the graph, the edge labels $a, b, c$ of any triangle $(x, y, z)$ (in fact, of any closed walk) must satisfy that $h(a) + h(b) + h(c) = 0$. By definition, $a, b, c$ constitutes a pseudo-solution.

For the other direction, let $a, b, c$ be a pseudo-solution. There are six ways to assign the edge labels to the edge parts; we will focus on one case and prove that there is a unique triangle $(x, y, z) \in X \times Y \times Z$ where $(x, y)$ is labeled with $a$, $(y, z)$ is labeled with $b$, and~$(z, x)$ is labeled with $c$. Writing $x = (x_1, x_2, 0)$, $y = (y_1, 0, y_3)$ and $z = (0, z_2, z_3)$, we obtain the following constraints:
\begin{align*}
    y_1 &= x_1 + h_1(a) & 0 &= x_2 + h_2(a) & y_3 &= h_3(a) \\
    0 &= y_1 + h_1(b) & z_2 &= h_2(b) & z_3 &= y_3 + h_3(b) \\
    x_1 &= h_1(c) & x_2 &= z_2 + h_2(c) & 0 &= z_3 + h_3(c)
\end{align*}
It is easy to check that this equation system (with indeterminates $x_1, x_2, y_1, y_3, z_2, z_3$) is uniquely solvable by $x_1 = h_1(c), x_2 = -h_2(a), y_1 = -h_1(b), y_3 = h_3(a), z_2 = h_2(b), z_3 = -h_3(c)$.
\end{proof}

By this characterization it is easy to complete the reduction: By listing all triangles in the constructed instance, in particular we list all pseudo-solutions of the 3-SUM instance. We check whether one of these pseudo-solutions forms a proper solution and return ``yes'' in this and only this case. Moreover, we obtain the following bound on the number of triangles in the constructed instance:

\begin{lemma}[Number of Triangles] \label{lem:number-of-triangles}
Either we can find a 3-SUM solution in time $\widetilde\Order(|G'|^3 / n)$, or the expected number of triangles in the constructed instance is $\Order(n^3 |G'|^{-3})$.
\end{lemma}
\begin{proof}
By the previous lemma, the number of triangles is bounded by six times the number of pseudo-solutions. We first focus on the pseudo-solutions $a, b, c$ which are not proper solutions (i.e.,~\makebox{$a + b + c \neq 0$}). The probability that $h(a) + h(b) + h(c) = 0$, or equivalently that $h(a + b + c) = 0$, is at most $|G'|^{-3}$. It follows that the expected number of non-proper pseudo-solutions is at most $n^3 |G'|^{-3}$.

Next, focus on the proper solutions. We distinguish two cases: On the one hand, if there are at most~$n^3 |G'|^{-3}$ proper solutions, then the total number of pseudo-solutions and therefore the total number of triangles is $\Order(n^3 |G'|^{-3})$, as claimed. On the other hand, if there are at least $n^3 |G'|^{-3}$ solutions, then it suffices to sample $\widetilde\Order(n^2 / (n^3 |G'|^{-3})) = \widetilde\Order(|G'|^3 / n)$ pairs $(a, b) \in A^2$ to detect at least one 3-SUM solution~$(a, b, -(a+b)) \in A^3$ with high probability.
\end{proof}

Finally, the instance can be constructed in time $\Order(n |G'|)$ as follows: We precompute the hash values $h_1(a), h_2(a), h_3(a)$ for all $a \in A$. For each vertex in the instance, say,~\makebox{$x = (x_1, x_2, 0)$}, we then check only those $a$'s with hash values satisfying $x_2 + h_2(a) = 0$ (or~$x_1 - h_1(a) = 0$) and add the respective edges.

\subsection{Counting the Number of \texorpdfstring{\boldmath$k$-Cycles}{k-Cycles}} \label{sec:3sum-to-triangle:sec:cycles}
The most interesting part in our setting is to bound the number of $k$-cycles in the constructed instance (for~$k \geq 4$). To this end, we introduce some notation. We say that a length-$k$ walk is \emph{labeled} by $a_1, \dots, a_k$ whenever the edges in the walk are labeled with $\pm a_1, \dots, \pm a_k$. More specifically, we fix an order of the vertex parts (say the \emph{clockwise} order is $X, Y, Z$) and require that the edge in the $i$-th step is labeled with~$a_i$ if the walk takes a step in clockwise direction (that is, from~$X$ to~$Y$, from~$Y$ to~$Z$ or from~$Z$ to~$X$) and labeled with $-a_i$ if the walk takes a step in counter-clockwise direction (that is, from $Y$ to $X$, $X$ to~$Z$ or from~$Z$ to~$Y$). For example, the walk~$a_1, -a_2, a_3, a_4$ for elements~$a_1, a_2, a_3, a_4 \in A$ takes one step in clockwise direction, takes one step in counter-clockwise direction (to the same part where it started from) and takes two more steps in clockwise direction. Here we assume for simplicity that $A$ and $-A$ are disjoint, so that the label of a walk uniquely determines its directions.\footnote{More generally, we should use pairs $(s_i, a_i)$ with $s_i = \pm 1$ and $a_i \in A$ to label paths, but we stick to the simpler version described in the text.}

We distinguish between two types of $k$-cycles: A $k$-cycle labeled with $a_1, \dots, a_k$ is called a \emph{pseudo-$k$-cycle} if $a_1 + \dots + a_k \neq 0$, and a \emph{zero-$k$-cycle} otherwise. The analysis differs for these two types of cycles: For pseudo-$k$-cycles we can exploit more randomness since all labels~$a_1, \dots, a_k$ can be expected to produce independent hash values $h(a_1), \dots, h(a_k)$. For zero-$k$-cycles, one of the hash values is determined by the others and we therefore have a smaller degree of independence. But we have the advantage that the 3-SUM instance has small additive energy, and therefore the number of solutions to $a_1 + \dots + a_k = 0$ is small.

\begin{lemma}[Rate of Zero-$k$-Cycles] \label{lem:rate-zero-cycles}
Fix a vertex $v$ and $a_1, \dots, a_k \in \pm A$ with $a_1 + \dots + a_k = 0$. Then there is a cycle starting from and ending at $v$ labeled with $a_1, \dots, a_k$ with probability at most $|G'|^{-s}$, where $s = \dim\Span{a_1, \dots, a_k}$.
\end{lemma}
\begin{proof}
First observe that any walk with labels $a_1 + \dots + a_k = 0$ that starts at $v$ also ends at~$v$. We therefore bound the probability that there is a walk starting from~$v$ which is labeled with $a_1, \dots, a_k$ by $|G'|^{-s}$. The proof is by induction on $k$. For the base $k = 0$ we have $s = 0$ and can trivially bound the probability by $1$.

For the inductive case assume that $k \geq 1$. By induction, there is a walk of length $k-1$ with probability at most $|G'|^{s'}$ where $s' = \dim\Span{a_1, \dots, a_{k-1}}$. We distinguish two cases: If~$s' = s$, then we are done. If~$s' = s - 1$ (which is indeed the only other case), then the vector $a_k$ is linearly independent from $a_1, \dots, a_{k-1}$ and thus the random variable $h(a_k)$ is independent from the other random variables $h(a_1), \dots, h(a_{k-1})$. Now suppose that the walk after~$k-1$ steps has reached some vertex, say, $x = (x_1, x_2, 0)$ and we move in clockwise direction. Then the target vertex $y = (y_1, 0, y_3)$ is uniquely determined by $y_1 = x_1 + h_1(a_k)$ and $y_3 = h_3(a_k)$. In addition, we induce the constraint $0 = x_2 + h_2(a_k)$ which is satisfied with probability at most $|G'|^{-1}$. By the aforementioned independence, the total probability is at most \raisebox{0pt}[0pt][0pt]{$|G'|^{-s'} |G'|^{-1} = |G'|^{-s}$}.
\end{proof}

\begin{lemma}[Rate of Pseudo-$k$-Cycles] \label{lem:rate-pseudo-cycles}
Fix a vertex $v$ and $a_1, \dots, a_k \in \pm A$ with $a_1 + \dots + a_k \neq 0$. Then there is a cycle starting from and ending at $v$ labeled with $a_1, \dots, a_k$ with probability at most $|G'|^{-s-2}$, where $s = \dim\Span{a_1, \dots, a_k}$.
\end{lemma}

The proof of this lemma is a bit more involved than the previous one. Our strategy is to prove the following more technical generalization (see \cref{lem:rate-pseudo-cycles-generalization}). The proof of \cref{lem:rate-pseudo-cycles} then follows by setting $a_0 = a_1 + \dots + a_k$. Indeed, any cycle labeled with $a_1, \dots, a_k$ is in particular a walk and because it is closed we must have $h(a_1) + \dots + h(a_k) = 0$.

\begin{lemma} \label{lem:rate-pseudo-cycles-generalization}
Fix a vertex $v$ and $a_1, \dots, a_k \in \pm A$ and any non-zero $a_0 \in G$. Then the probability of the combined events that (i) there is a walk starting from $v$ labeled with $a_1, \dots, a_k$ and (ii) $h(a_0) = 0$, is at most $|G'|^{-s-2}$, where $s = \dim\Span{a_0, a_1, \dots, a_k}$.
\end{lemma}
\begin{proof}
The proof is by induction on $k$. We start with the base case $k = 0$. Since we assume that $a_0 \neq 0$, we have that $s = \Span{a_0} = 1$. Moreover, the probability that $h(a_0) = 0$ is exactly~$|G'|^{-3}$.

Next consider the inductive case $k \geq 1$, and let $s' = \dim\Span{a_0, \dots, a_{k-1}}$. If $s' = s$, then we are done by induction. Otherwise, we have $s' = s - 1$ and $a_k$ is linearly independent from the other vectors~$a_0, \dots, a_{k-1}$. Suppose that after $k-1$ steps the walk has reached some vertex, say $z = (0, z_2, z_3)$, and we are moving in counter-clockwise direction. Then the target vertex $y = (y_1, 0, y_3)$ is uniquely determined by $y_1 = h_1(a_k)$ and $y_3 = z_3 + h_3(a_k)$, but moving to $y$ is only possible if the new constraint $0 = z_2 + h_2(a_k)$ is satisfied. This constraint is satisfied with probability $|G'|^{-1}$ and since $h(a_k)$ is independent from the randomness in previous steps, the overall probability is at most $|G'|^{-s-1} \cdot |G'|^{-1} \leq |G'|^{-s-2}$.
\end{proof}

\begin{lemma}[Number of $k$-Cycles] \label{lem:number-of-cycles}
For any constant $k \geq 4$, the expected number of $k$-cycles in the constructed instance is $\Order(E(A) \cdot n^{k-4} |G'|^{-k+3} + n^{k-2} |G'|^{-k+4} + n^k |G'|^{-k})$.
\end{lemma}
\begin{proof}
We first compute the expected number of pseudo-$k$-cycles.
\begin{gather*}
    \sum_{\substack{a_1, \dots, a_k \in \pm A\\a_1 + \dots + a_k \neq 0}}\sum_{v \in V} \Pr(\text{there is a cycle starting from and ending at $v$ labeled with $a_1, \dots, a_k$}) \\
    \qquad\leq \sum_{s=0}^k \sum_{\substack{a_1, \dots, a_k \in \pm A\\a_1 + \dots + a_k \neq 0\\\dim\Span{a_1, \dots, a_k} = s}}\sum_{v \in V} |G'|^{-s-2} \\
    \qquad\leq \sum_{s=0}^k \Order(n^s \cdot |G'|^2 \cdot |G'|^{-s-2}) \\
    \qquad= \Order(n^k |G'|^{-k}).
\end{gather*}
Here, for the first inequality we have applied \cref{lem:rate-pseudo-cycles} and we have bounded the number of tuples $(a_1, \dots, a_k)$ with $\dim\Span{a_1, \dots, a_k} = s$ by $\Order(n^s)$ (indeed, after fixing $s$ linearly independent vectors from~$\pm A$, each remaining vector can be expressed as one out of $p^k \leq \Order(1)$ possible linear combinations).

Next, we compute the number of zero-$k$-cycles:
\begin{gather*}
    \sum_{\substack{a_1, \dots, a_k \in \pm A\\a_1 + \dots + a_k = 0}}\sum_{v \in V} \Pr(\text{there is a cycle starting from and ending at $v$ labeled with $a_1, \dots, a_k$}) \\
    \qquad\leq \sum_{s=0}^{k-1} \sum_{\substack{a_1, \dots, a_k \in \pm A\\a_1 + \dots + a_k = 0\\\dim\Span{a_1, \dots, a_k} = s}}\sum_{v \in V} |G'|^{-s} \\
    \qquad= \sum_{\substack{a_1, \dots, a_k \in \pm A\\a_1 + \dots + a_k = 0\\\dim\Span{a_1, \dots, a_k} = k-1}}\sum_{v \in V} |G'|^{-k+1} + \sum_{s=0}^{k-2} \sum_{\substack{a_1, \dots, a_k \in \pm A\\a_1 + \dots + a_k = 0\\\dim\Span{a_1, \dots, a_k} = s}}\sum_{v \in V} |G'|^{-s} \\
    \qquad\leq \Order(E(A) \cdot n^{k-4} \cdot |G'|^2 \cdot |G'|^{-k+1}) + \sum_{s=0}^{k-2} \Order(n^s \cdot |G'|^2 \cdot |G'|^{-s}) \\
    \qquad= \Order(E(A) \cdot n^{k-4} |G'|^{-k+3} + n^{k-2} |G'|^{-k+4}).
\end{gather*}
For the first inequality we applied \cref{lem:rate-zero-cycles}. For the second inequality, we have bounded the number of tuples $a_1, \dots, a_k$ with $\dim\Span{a_1, \dots, a_k} = s$ by $\Order(n^s)$ as before. In addition, we have bounded the number of solutions $a_1, \dots, a_k \in \pm A$ to the linear equation~$a_1 + \dots + a_k = 0$ using \cref{lem:energy-linear-equations} by $E(A) \cdot |A|^{k-4}$.
\end{proof}

\subsection{Making the Graph Regular} \label{sec:3sum-to-triangle:sec:regular}
The next step is to enforce the assumption that the constructed is $\Theta(r)$-regular, where $r = 2n / |G'|$. To this end, we first analyze the \emph{expected degrees} in the instance constructed in the previous \cref{sec:3sum-to-triangle:sec:construction}.

\begin{lemma}
Fix a vertex $v$. Then $\Ex(\deg(v)) = r \pm \Order(1)$ and $\Var(\deg(v)) \leq \Order(r)$.
\end{lemma}
\begin{proof}
Focus on an arbitrary vertex, say, $x = (x_1, x_2, 0) \in X$ (the proof is similar for vertices in $Y$ and $Z$). We write $\deg(x) = \deg_Y(x) + \deg_Z(x)$, where $\deg_Y(x)$ denotes the number of edges from $x$ to $Y$, and $\deg_Z(x)$ denotes the number of edges from $x$ to $Z$. We focus on the analysis of $\deg_Y(x)$, the same treatment applies to $\deg_Z(x)$. For each edge label $a$, there is only a unique candidate $y = (y_1, 0, y_3) \in Y$ which is reachable by an edge from $x$ (indeed, $y_1$ and $y_3$ are determined by $x_1$ and $a$). There is an edge to that unique candidate~$y$ if and only if $x_2 + h_2(a) = 0$. Therefore, the expected degree is:
\begin{equation*}
    \Ex(\deg_Y(x)) = \sum_{a \in A} \Pr(x_2 + h_2(a) = 0) = n |G'|^{-1} \pm \Order(1).
\end{equation*}
(The $\pm \Order(1)$ term stems from the element $0$ which may or may not be present in $A$ but always hashes to $0$ under a linear hash function.) To bound the variance, we compute
\begin{gather*}
    \Var(\deg_Y(x)) \\
    \qquad= \Ex(\deg_Y(x)^2) - \Ex(\deg_Y(x))^2 \\
    \qquad\leq \left(\sum_{a, b \in A} \Pr(x_2 + h_2(a) = x_2 + h_2(b) = 0)\right) - \left(\sum_{a \in A} \Pr(x_2 + h_2(a) = 0)\right)^2
\intertext{To bound the first sum, we consider two cases: Either $a$ and $b$ are linearly independent, in which case the random variables $h_2(a)$ and $h_2(b)$ are independent. Or $a$ and $b$ are linearly dependent, in which case there are at most $p n = \Order(n)$ such pairs (we can pick $a$ arbitrarily and there are only $p$ choices for $b$ in the span~$\Span{a}$). It follows that:}
    \qquad\leq \Order(n|G'|^{-1}) + \left(\sum_{a, b \in A} \Pr(x_2 + h_2(a) = 0) \cdot \Pr(x_2 + h_2(b) = 0)\right) - \left(\sum_{a \in A} \Pr(x_2 + h_2(a) = 0)\right)^2 \\
    \qquad= \Order(n|G'|^{-1}).
\end{gather*}
Recall that $\deg(x) = \deg_Y(x) + \deg_Z(x)$. Since the random variables $\deg_Y(x)$ and $\deg_Z(x)$ depend on the independent hash functions $h_2$ and $h_1$, the random variables $\deg_Y(x)$ and $\deg_Z(x)$ are independent. It follows that $\Ex(\deg(x))$ and $\Var(\deg(x))$ are as claimed.  
\end{proof}

Given the previous lemma, most vertices in the constructed instance have degree $\Theta(r)$. However, we want that \emph{every} vertex has degree $\Theta(r)$. We will therefore select an (induced) subgraph of the constructed instance, in which the degree bound is satisfied. Note that by selecting a subgraph, we cannot increase the number of $k$-cycles, and the analysis from the previous \cref{sec:3sum-to-triangle:sec:cycles} remains intact.

We use the algorithm described in \cref{alg:regular}. It is easiest to describe using some terminology: We call a vertex $v$ \emph{high-degree} if it has degree more than $2r$, \emph{low-degree} if it has degree less than $\frac12 r$ and \emph{tiny-degree} if it has degree less than $\frac18 r$. As long as there is a high-degree or tiny-degree vertex $v$ in the graph, we remove $v$ and all its incident edges. In order to not miss the triangles involving the removed vertices $v$, we list all pairs of neighbors $u, w$ of $v$ and report all triangles $(u, v, w)$ found in this way. It is obvious that the remaining graph is $\Theta(r)$-regular, and moreover we have not missed any triangle by pruning the graph in this way. It remains to analyze the running time of \cref{alg:regular}. 

\begin{algorithm}[t]
\caption{Turns the triangle listing instance from \cref{sec:3sum-to-triangle:sec:construction} into a $\Theta(n / |G|)$-regular graph (by removing some vertices, and listing all triangles involving at least one of the removed vertices).} \label{alg:regular}
\begin{algorithmic}[1]
\State Let $(V_0, E_0)$ be the instance constructed in \cref{sec:3sum-to-triangle:sec:construction}
\State Let $V \gets V_0, E \gets E_0$ and let $r \gets 2n / |G'|$
\While{there is a vertex $v$ in $(V, E)$ with degree less than $\frac14 r$ or more than $2r$}
    \State Enumerate all pairs of neighbors $u, w \in V$ of $v$ and report $(u, v, w)$ if it is a triangle
    \State Remove $v$ from $V$ and its incident edges from $E$
\EndWhile
\State\Return $(V, E)$
\end{algorithmic}
\end{algorithm}

\begin{lemma}[Running Time of \cref{alg:regular}] \label{lem:regular-time}
\cref{alg:regular} runs in expected time $\widetilde\Order(n |G'|)$.
\end{lemma}
\begin{proof}
As in \cref{alg:regular}, we denote by $(V_0, E_0)$ the graph constructed in \cref{sec:3sum-to-triangle:sec:construction}. We split the analysis in two parts: First, we bound the time spend in iterations removing a high-degree vertex and second, we bound the time spend in iterations removing tiny-degree vertices. Since the running time per iteration is dominated by enumerating all pairs of neighbors of the vertex $v$ to be removed, we can bound the expected time to remove all high-degree vertices as follows:
\begin{gather*}
    \sum_{v \in V_0} \deg(v)^2 \cdot \Pr(\deg(v) \geq 2r) \\
    \qquad\leq \sum_{v \in V_0} \sum_{i=1}^{\log |V_0|} 2^{2i+2} r^2 \cdot \Pr(\deg(v) \geq 2^i r) \\
    \qquad\leq \sum_{v \in V_0} \sum_{i=1}^{\log |V_0|} 2^{2i+2} r^2 \cdot \Pr(|\deg(v) - \Ex(\deg(v))| \geq \Omega(2^i \cdot \Var(\deg(v)))) \\
    \qquad\leq \sum_{v \in V_0} \sum_{i=1}^{\log |V_0|} 2^{2i+2} r^2 \cdot \Order\left(\frac{1}{2^{2i} r}\right) \\
    \qquad\leq \sum_{v \in V_0} \sum_{i=1}^{\log |V_0|} \Order(r) \\
    \qquad\leq \widetilde\Order(n |G'|).
\end{gather*}

We now focus on the time spent on iterations removing tiny-degree vertices. Each such iteration runs in time $\Order(r^2)$, and we therefore aim to bound the number of iterations. The first step is to show that in the original graph $(V_0, E_0)$, the expected number of edges incident to high-degree or low-degree vertices is at most $\Order(n)$. Indeed, by Chebyshev's inequality and again using the previously obtained bounds, the expected number of edges incident to high-degree vertices is at most
\begin{gather*}
    \sum_{v \in V_0} \deg(v) \cdot \Pr\left(\deg(v) \geq 2r\right) \\
    \qquad\leq \sum_{v \in V_0} \sum_{i=1}^\infty 2^{i+1} r \cdot \Pr(\deg(v) \geq 2^{i} r) \\
    \qquad\leq \sum_{v \in V_0} \sum_{i=1}^\infty 2^{i+1} r \cdot \Pr(|\deg(v) - \Ex(\deg(v))| \geq \Omega(2^i \cdot \Var(\deg(v)))) \\
    \qquad\leq \sum_{v \in V_0} \sum_{i=1}^\infty 2^{i+1} r \cdot \Order\left(\frac{1}{2^{2i} r}\right) \\
    \qquad\leq \sum_{v \in V_0} \Order(1) \\
    \qquad= \Order(|G'|^2). 
\end{gather*}
Using the same idea we can bound the number of edges incident to low-degree vertices by $\Order(|G'|^2)$, too. Moreover, we can bound the numbers $L$ and $H$ of low-degree and high-degree vertices in the original graph by $L, H = \Order(|G'|^2 / r)$.

We now again turn to \cref{alg:regular} and bound the number of iterations. There are up to $H$ iterations removing the high-degree vertices, and the remaining iterations remove tiny-degree vertices. However, observe after removing $e$ edges from the original graph, we can create at most~$L + 6e / r$ tiny-degree vertices: Up to $L$ vertices which are low-degree in the original graph plus at most $2e / (\frac12 r - \frac18 r) \leq 6e / r$ vertices which were not low-degree in the original graph but which turned tiny-degree by losing edges. Since every iteration removing a tiny-degree vertex removes at most $\frac14 r$ edges, the total number of edges removed after $i$ iterations is at most $\Order(|G'|^2) + \frac i4 r$. Consequently, if the algorithm reaches the $i$-th iteration, it has witnessed at least~$i - H$ tiny-degree vertices and we therefore have
\begin{equation*}
    i - H \leq L + \frac{6 \cdot (\Order(|G'|^2) + \frac{i}8 r)}{r} \leq L + \Order(|G'|^2 / r) + \frac{3i}{4}.
\end{equation*}
It follows that $i \leq \Order(L + H + |G'|^2 / r) = \Order(|G'|^2 / r)$, and therefore \cref{alg:regular} runs for at most $\Order(|G'|^2 / r)$ iterations. Recall that each iteration removing a tiny-degree vertex takes time $\Order(r^2)$, and therefore the total time of all iterations removing tiny-degree vertices is $\Order(|G'|^2 / r \cdot r^2) = \Order(n |G'|)$.
\end{proof}

\subsection{Putting the Pieces Together} \label{sec:3sum-to-triangle:sec:assembling}

We are ready to prove \cref{lem:3sum-to-triangle}.

\begin{proof}[Proof of \cref{lem:3sum-to-triangle}]
Recall that we start from a 3-SUM instance with additive energy $E(A) \leq \Order(n^{5/2})$, by the energy reduction in \cref{lem:energy-reduction-complete} applied with $\delta = \frac12$. We set $|G'| = n^{1/2}$ (that is, we set \makebox{$d' = \ceil{\frac12 \log_p(n)}$} and \raisebox{0pt}[0pt][0pt]{$G' = \Field_p^{d'}$}) and construct the triangle listing instance $(V_0, E_0)$ as described in \cref{sec:3sum-to-triangle:sec:construction}. This step takes time~$\Order(n |G'|) = \Order(n^{3/2})$. We then run \cref{alg:regular} as described in \cref{sec:3sum-to-triangle:sec:regular} to obtain an induced subgraph $(V_1, E_1)$ which is regular with degree $\Theta(n / |G'|) = \Theta(n^{1/2})$. This step again takes time \raisebox{0pt}[0pt][0pt]{$\widetilde\Order(n |G'|) = \Order(n^{3/2})$} in expectation, see \cref{lem:regular-time}.

We next bound the (expected) number of $k$-cycles, for any $3 \leq k \leq k_{\max}$. By \cref{lem:number-of-triangles}, the expected number of triangles in $(V_0, E_0)$ is at most $\Order(n^3 |G'|^{-3})$ (alternatively, we can immediately find a 3-SUM solution in time~\raisebox{0pt}[0pt][0pt]{$\widetilde\Order(|G'|^3 / n) = \widetilde\Order(n^{1/2})$}). For $k \geq 4$, by \cref{lem:number-of-cycles} the expected number of $k$-cycles in~$(V_0, E_0)$ is at most
\begin{gather*}
    \Order(E(A) \cdot n^{k-4} |G'|^{-k+3} + n^{k-2} |G'|^{-k+4} + n^k |G'|^{-k}) \\
    \qquad= \Order(n^{k-3/2} n^{-k/2+3/2} + n^{k-2} n^{-k/2+2} + n^k n^{-k/2}) \\
    \qquad= \Order(n^{k/2}).
\end{gather*}
Using Markov's bound, this number exceeds $10 k_{\max}$ times its expected value with probability at most \raisebox{0pt}[0pt][0pt]{$\frac1{10k_{\max}}$}. Therefore, taking a union bound over all values $3 \leq k \leq k_{\max}$, the constructed instances $(V_0, E_0)$ and~$(V_1, E_1)$ contain at most $\Order(n^{k/2})$ $k$-cycles for all $3 \leq k \leq k_{\max}$, with probability at least~$\frac9{10}$.

Now suppose that we can list all triangles in $(V_1, E_1)$ in time $\Order(n^{2-\epsilon})$. Adding the triangles detected by \cref{alg:regular}, we can compute a list of all triangles in $(V_0, E_0)$. Recall that by \cref{lem:pseudo-solutions}, every triangle corresponds to a pseudo-solution in the 3-SUM instance. Therefore, it suffices to test whether there exists a proper solution among the pseudo-solutions and to return ``yes'' in this case. The total expected running time is $\widetilde\Order(n^{3/2} + n^{2-\epsilon})$ and we succeed with constant error probability.
\end{proof}

\paragraph{Listing Hardness in Graphs with Smaller Degrees}
For one of our corollaries of the reduction we need denser graphs than the $\Theta(n^{1/2})$-regular graphs constructed before. It is easy to obtain the following generalization of our reduction to graphs which are $\Theta(r)$-regular.

\begin{lemma}[Hardness of Listing Triangles in $\Theta(r)$-Regular Graphs] \label{lem:3sum-to-triangle-degree}
For any $\epsilon > 0$ and any parameter $N^{1/2} \leq r \leq N^{1-\Omega(1)}$, there is no $\Order((Nr^2)^{1-\epsilon})$-time algorithm listing all triangles in a $\Theta(r)
$-regular $N$-vertex graph which contains as most $\Order(r^k)$ $k$-cycles for all~$3 \leq k \leq \Order(1)$, unless the 3-SUM conjecture fails.
\end{lemma}
\begin{proof}
We redo the proof of \cref{lem:3sum-to-triangle} with a different choice of parameters. Specifically, start from a 3-SUM instance of size $n = N^{1/2} r$ and with additive energy $E(A) \leq \Order(n^{5/2})$ and set $|G'| = N^{1/2}$. The constructions in \cref{sec:3sum-to-triangle:sec:construction,sec:3sum-to-triangle:sec:regular} construct a graph with at most~$|G'|^2 = N$ vertices, and the degree of every vertex is $\Theta(n / |G'|) = \Theta(r)$. The running time of these steps is bounded by $\widetilde\Order(n |G'| + |G'|^3 / n)$ (by \cref{lem:number-of-triangles,lem:regular-time}). By \cref{lem:number-of-triangles} the expected number of triangles is $\Order(n^3 |G'|^{-3}) = \Order(r^3)$ and by \cref{lem:number-of-cycles}, the expected number of $k$-cycles is bounded
\begin{gather*}
    \Order(E(A) \cdot n^{k-4} |G'|^{-k+3} + n^{k-2} |G'|^{-k+4} + n^k |G'|^{-k}) \\
    \qquad= \Order(N^{k/2-3/4} r^{k-3/2} N^{-k/2+3/2} + N^{k/2-1} r^{k-1} N^{-k/2+2} + N^{k/2} r^k N^{-k/2}) \\
    \qquad= \Order(N^{3/4} r^{k-3/2} + N r^{k-2} + r^k) \\
    \qquad\leq \Order(r^k).
\end{gather*}
For last step we have used the assumption $N^{1/2} \leq r$. Finally, an algorithm in time $\Order((Nr^2)^{1-\epsilon})$ would imply an algorithm in time~$\Order(n^{2-2\epsilon} + n |G'| + |G'|^3 / n)$ for the 3-SUM instance we started from. As $n^{1/2} \leq |G'| \leq n^{1-\Omega(1)}$, this is subquadratic and contradicts the 3-SUM conjecture.
\end{proof}

\paragraph{All-Edges Triangle}
Many reductions starting from triangle listing can be phrased in a nicer way by starting instead from the \emph{All-Edges Triangle} problem: Given a graph, determine for each edge whether it is part of a triangle. Using our reduction and in addition some known tricks to turn detection algorithms into witness-finding algorithms, we also obtain the following conditional lower bound:

\begin{lemma}[Hardness of All-Edges Triangle]
For any constants $\epsilon > 0, k_{\max} \geq 3$, there is no $\Order(n^{2-\epsilon})$-time algorithm for the All-Edges Triangle problem in $\Theta(n^{1/2})$-regular $n$-vertex graphs which contain at most~$\Order(n^{k/2})$ $k$-cycles for all $3 \leq k \leq k_{\max}$, unless the 3-SUM conjecture fails.
\end{lemma}
\section{Hardness of 4-Cycle Listing} \label{sec:4cycle}

\begin{algorithm}[t]
\caption{The reduction from listing triangles in a $\Theta(n^{1/2})$-regular tripartite graph $G = (V, E)$ to listing 4-cycles.} \label{alg:4cycle}
\begin{algorithmic}[1]
\State Randomly split $V$ into $V_1, \dots, V_s$
\ForEach{$(i, j, \ell) \in [s]^3$}
    \State Let $V_{i, j, \ell} = V_i \cup V_j \cup V_\ell$
    \State Let $G_{i, j, \ell}$ be the graph with vertices $\set{x_1, x_2, x_3, x_4 : x \in V_{i, j, \ell}}$ (that is, create four
    \Statex[1] copies for each original vertex) and edges $\set{(x_1, y_2), (x_2, y_3), (x_3, y_4) : (x, y) \in E}$
    \Statex[1] and $\set{(x_1, x_4) : x \in V_{i, j, \ell}}$
    \State Run the fast 4-cycle listing algorithm on $G_{i, j, \ell}$, and for each 4-cycle of the form
    \Statex[1] $(x_1, y_2, z_3, x_4)$ report the triangle $(x, y, z)$ (unless already reported)
\EndForEach
\end{algorithmic}
\end{algorithm}

\thmhardnessfourcycle*

\noindent
This section is devoted to proving \cref{thm:hardness-4cycle}. Suppose that for some $\epsilon > 0$, there is an algorithm listing all 4-cycles in a graph in time $\Order(n^{2-\epsilon} + t)$. We give a reduction from listing triangles as described in \cref{lem:3sum-to-triangle} (with $k_{\max} = 4$) to listing 4-cycles. That is, we are given an $\Theta(n^{1/2})$-regular $n$-vertex graph~$G = (V, E)$ which contains at most $\Order(n^2)$ 4-cycles, and the goal is to list $\Order(n^{3/2})$ triangles in subquadratic time. The reduction is summarized in \cref{alg:4cycle}.

The algorithm randomly splits the vertex set into $s$ groups, and for each triple $(i, j, \ell) \in [s]^3$ of groups, constructs a new graph $G_{i, j, \ell}$. This graph is obtained from $G$ by copying each vertex $x$ four times $x_1, x_2, x_3, x_4$, and we add edges $(x_1, y_2), (x_2, y_3), (x_3, y_4)$ as in the original graph, and additionally add all edges $(x_1, x_4)$. We list all 4-cycles in the graph~$G_{i, j, \ell}$ and for each 4-cycle of the form $(x_1, y_2, z_3, x_4)$ we report the triangle $(x, y, z)$. Our first claim is that the algorithm correctly reports all triangles in $G$.

\begin{lemma}[Correctness of \cref{alg:4cycle}] \label{lem:4cycle-correctness}
\cref{alg:4cycle} correctly lists all triangles in $G$.
\end{lemma}
\begin{proof}
First, observe that by the construction of $G_{i, j, \ell}$ every triple $(x, y, z)$ reported by the algorithm indeed forms a triangle in $G$. Moreover, any triangle $(x, y, z)$ in $G$ can be found as the 4-cycle $(x_1, y_2, z_3, x_4)$ in $G_{i, j, \ell}$, where $x \in V_i, y \in V_j, z \in Z_\ell$. (In addition, there are five other 4-cycles which correspond to $(x, y, z)$.)
\end{proof}

\begin{lemma}[Number of 4-Cycles] \label{lem:number-of-4cycles}
The expected total number of 4-cycles across all graphs~$G_{i, j, \ell}$ is at most $\Order(n^2 / s + n^{3/2})$.
\end{lemma}
\begin{proof}
Each 4-cycle using an edge $(x_1, x_4)$ must take the form $(x_1, y_2, z_3, x_4)$. In this case, $(x, y, z)$ is a triangle in the original graph $G$. As each triangle appears as a four cycle in all six possible permutations, the contribution from 4-cycles using an edge $(x_1, x_4)$ is therefore bounded by six times the number of triangles in $G$. By \cref{lem:3sum-to-triangle}, $G$ contains at most $\Order(n^{3/2})$ many triangles.

There are five types of 4-cycles which do not use edges of the form $(x_1, x_4)$, namely $(x_1, y_2, z_1, w_2)$, $(x_2, y_3, z_2, w_3)$, $(x_3, y_4, z_3, w_4)$, $(x_1, y_2, z_3, w_2)$ and $(x_2, y_3, z_4, w_3)$. In all five cases, $(x, y, z, w)$ forms a 4-cycle in the original graph---more specifically, in the subgraph induced by $V_{i, j, \ell}$. In particular, the contribution of these 4-cycles is five times the number of 4-cycles in $G[V_{i, j, \ell}]$. Recall that in $G$ there are only~$\Order(n^2)$ 4-cycles, and each 4-cycle \emph{survives} only if all of its four vertices are sampled into some set $V_{i, j, \ell}$ which happens with probability~$s^3 \cdot s^{-4} = s^{-1}$. Hence, the expected total number of surviving 4-cycles is~$\Order(n^2 / s)$.
\end{proof}

\begin{lemma}[Running Time of \cref{alg:4cycle}] \label{lem:4cycle-time}
For $s = n^{\epsilon/4}$, \cref{alg:4cycle} runs in expected time $\Order(n^{2-\epsilon/4})$.
\end{lemma}
\begin{proof}
The total running time is dominated by the running time of the fast 4-cycle listing algorithm. Assume that this algorithm runs in time $\widetilde\Order(n^{2-\epsilon} + t_{i, j, \ell})$ where $t_{i, j, \ell}$ is the number of 4-cycles in the respective instance. By the previous \cref{lem:number-of-4cycles} we have that~$\sum_{i, j, \ell} t_{i, j, \ell} \leq \Order(n^2 / s) = \Order(n^{2-\epsilon/4})$. Hence, the total running time of \cref{alg:4cycle} is
\begin{equation*}
    \sum_{i, j, \ell \in [s]} \widetilde\Order(n^{2-\epsilon} + t_{i, j, \ell}) \leq \widetilde\Order(n^{3\epsilon/4} \cdot n^{2-\epsilon}) + \widetilde\Order(n^{2-\epsilon/4}) = \widetilde\Order(n^{2-\epsilon/4}).
\end{equation*}
Similarly, if the fast 4-cycle listing algorithm runs in time $\widetilde\Order(m^{4/3-\epsilon} + t_{i, j, \ell})$, then the running time becomes
\begin{equation*}
    \sum_{i, j, \ell \in [s]} \widetilde\Order((n^{3/2})^{4/3-\epsilon} + t_{i, j, \ell}) \leq \widetilde\Order(n^{3\epsilon/4} \cdot n^{2-\epsilon}) + \widetilde\Order(n^{2-\epsilon/4}) = \widetilde\Order(n^{2-\epsilon/4}). \qedhere
\end{equation*}
\end{proof}

The proof of \cref{thm:hardness-4cycle} is complete by \cref{lem:4cycle-correctness,lem:4cycle-time}. If necessary we can further let the algorithm terminate with high probability in time $\widetilde\Order(n^{2-\epsilon/4})$ by repeating the reduction~$\Order(\log n)$ times and interrupting each execution which takes too long.
\section{Hardness of Distance Oracles} \label{sec:distance-oracles}

In this section we prove our conditional hardness results for approximate distance oracles. We start with the stretch-$k$ regime (in \cref{sec:distance-oracles:sec:large}), followed by the stretch-$\alpha$ regime for~\makebox{$2 \leq \alpha < 3$} (in \cref{sec:distance-oracles:sec:small}), and the improved hardness for dynamic approximate distance oracles (in \cref{sec:distance-oracles:sec:dynamic}).

\subsection{Stretch \texorpdfstring{\boldmath$k$}{k}} \label{sec:distance-oracles:sec:large}
The goal of this section is to prove the following theorem:
\thmhardnessdistanceoraclelarge*

\begin{algorithm}[t]
\caption{The reduction from listing triangles in a $\Theta(n^{1/2})$-regular tripartite graph $G = (X, Y, Z, E)$ to approximate distance oracles with stretch $k$.} \label{alg:3sum-to-distance-oracle}
\begin{algorithmic}[1]
\State Randomly split $X, Y, Z$ into $X_1, \dots, X_s$, $Y_1, \dots, Y_t$, $Z_1, \dots, Z_s$ \label{alg:3sum-to-distance-oracle:line:split}
\ForEach{$(i, j, \ell) \in [s] \times [t] \times [s]$} \label{alg:3sum-to-distance-oracle:line:loop-split}
    \State Let $G_{i, j, \ell}$ be the subgraph of $G$ induced by $X_i, Y_j, Z_\ell$ where all edges between $X_i$ \label{alg:3sum-to-distance-oracle:line:construct}
    \Statex[1] and $Z_\ell$ are deleted
    \State Preprocess $G_{i, j, \ell}$ with the approximate distance oracle \label{alg:3sum-to-distance-oracle:line:preprocess}
    \ForEach{$(x, z) \in (X_i \times Z_\ell) \cap E$} \label{alg:3sum-to-distance-oracle:line:loop-query}
        \State Query the distance oracle to obtain an estimate $d(x, z) \leq \widetilde d(x, z) \leq k \cdot d(x, z)$ \label{alg:3sum-to-distance-oracle:line:query}
        \If{$\widetilde d(x, z) \leq 2k$} \label{alg:3sum-to-distance-oracle:line:close-condition}
            \ForEach{$y \in Y_j$ with $(x, y) \in E$} \label{alg:3sum-to-distance-oracle:line:loop-test}
                \If{$(y, z) \in E$} \label{alg:3sum-to-distance-oracle:line:test}
                    \State Report the triangle $(x, y, z)$ \label{alg:3sum-to-distance-oracle:line:report}
                \EndIf
            \EndForEach
        \EndIf
    \EndForEach
\EndForEach
\end{algorithmic}
\end{algorithm}

For the remainder of this subsection, we prove \cref{thm:hardness-distance-oracles}. Assume that we have access to an approximate distance oracle with stretch $k$, preprocessing time $\widetilde\Order(m^{1+p})$ and query time $\widetilde\Order(m^q)$.

We prove hardness starting from an instance of listing $\Order(n^{3/2})$ triangles in a $\Theta(n^{1/2})$-regular tripartite $n$-vertex graph $G = (X, Y, Z, E)$ which contains at most $\Order(n^{k'/2})$ $k'$-cycles for all \makebox{$4 \leq k'\leq 2k+1$} (that is, we apply the hardness result from \cref{lem:3sum-to-triangle} with~\makebox{$k_{\max} = 2k+1$}, and the additional assumption that $G$ be tripartite is without loss of generality). We let $s, t \leq n^{1/2-\Omega(1)}$ be two parameters to be set later and give the reduction in \cref{alg:3sum-to-distance-oracle}.

The algorithm first splits the vertex parts $X, Y, Z$ into $s, t, s$ many groups $X_i, Y_j, Z_\ell$, respectively, and then considers all graphs $G_{i, j, \ell}$ induced by $X_i \cup Y_j \cup Z_\ell$, where we have deleted all edges between $X_i$ and $Z_\ell$. We precompute $G_{i, j, \ell}$ with the distance oracle, and query the oracle for estimates $d(x, z) \leq \widetilde d(x, z) \leq k \cdot d(x, z)$ for all pairs $(x, z) \in (X_i \times Z_\ell) \cap E$. We call a pair $(x, z)$ with estimate $\widetilde d(x, z) \leq 2k$ a \emph{candidate} pair. The algorithm enumerates all candidate pairs $(x, z)$ and all neighbors $y$ of $x$, tests whether $(x, y, z)$ forms a triangle (in the original graph) and reports the triangle in the positive case. It is easy to see that the reduction is correct:

\begin{lemma}[Correctness of \cref{alg:3sum-to-distance-oracle}]
The reduction in \cref{alg:3sum-to-distance-oracle} correctly lists all triangles in the given graph $G = (X, Y, Z, E)$.
\end{lemma}
\begin{proof}
First note that whenever the algorithm reports a triangle $(x, y, z)$, we have verified  in \cref{alg:3sum-to-distance-oracle:line:loop-query,alg:3sum-to-distance-oracle:line:loop-test,alg:3sum-to-distance-oracle:line:test} that all edges $(x, y), (y, z), (x, z)$ are present.

Next, focus on any triangle $(x, y, z)$ in $G$; we prove that it is reported by the algorithm. Clearly there exist $i \in [s], j \in [t], \ell \in [s]$ such that $x \in X_i, y \in Y_j, z \in Z_\ell$. Focus on the iteration of the loop in \cref{alg:3sum-to-distance-oracle:line:loop-split} with~$(i, j, \ell)$ and on the iteration of the inner loop in \cref{alg:3sum-to-distance-oracle:line:loop-query} with $(x, z)$. The distance oracle is queried to obtain a distance estimate $\widetilde d(x, z) \leq k \cdot d(x, z)$ for the distance of $x$ and $z$ in $G_{i, j, \ell}$. Note that $x$ and $z$ are connected by a 2-path via $y$, hence the distance estimate satisfies $\widetilde d(x, z) \leq 2k$ (that is, $(x, z)$ is indeed a candidate pair). It follows that we enter the loop in \cref{alg:3sum-to-distance-oracle:line:loop-test} and report $(x, y, z)$ in \cref{alg:3sum-to-distance-oracle:line:report}.
\end{proof}

The more interesting part of the proof is to bound the running time of the reduction. For the analysis, we first analyze the sizes and degrees of the graphs $G_{i, j, \ell}$. It is easy to see that all bounds are true in expectation, and the high probability bounds follow from Chernoff's bound.

\begin{lemma}[Size of $G_{i, j, \ell}$]
With high probability the following bounds hold for all $(i, j, \ell) \in [s] \times [t] \times [s]$:
\begin{itemize}
\item $|X_i|, |Z_\ell| \leq \Order(n/s)$ and $|Y_j| \leq \Order(n/t)$.
\item $|(X_i \times Y_j) \cap E|, |(Y_j \times Z_\ell) \cap E| \leq \Order(n^{3/2} / st)$ and $|(X_i \times Z_\ell) \cap E| \leq \Order(n^{3/2} / s^2)$.
\item The degree of any vertex $x \in X_i$ in $G_{i, j, \ell}$ is $\Order(n^{1/2} / t)$.
\end{itemize}
\end{lemma}

\begin{lemma}[Few Candidates] \label{lem:few-short-paths}
Fix $i, j, \ell \in [s] \times [t] \times [s]$. In expectation, the number of candidate pairs~$(x, z) \in (X_i \times Z_\ell) \cap E$ is at most
\begin{equation*}
    \widetilde\Order\left(\frac{n^{k+1/2}}{s^{k+1} t^k}\right).
\end{equation*}
\end{lemma}
\begin{proof}
Since each candidate pair $(x, z)$ has distance $d(x, z) \leq 2k$ in $G_{i, j, \ell}$, $x$ and $z$ must be connected by a path of length $2k' \leq 2k$ in $G_{i, j, \ell}$. It follows that $(x, z)$ is part of a cycle of (odd) length $2k' + 1 \leq 2k+1$ in the induced subgraph $G[X_i \cup Y_j \cup Z_\ell]$. So fix any cycle in~$G$ of length $2k' + 1 \leq 2k+1$ which uses exactly one edge between $X$ and $Z$. In this case the cycle has exactly $k'+1$ vertices in $X \cup Z$ and exactly $k'$ vertices in~$Y$. The probability that this cycle is also contained in~$G[X_i \cup Y_j \cup Z_\ell]$ is therefore at most $(1/s)^{k'+1} (1/t)^{k'}$. Since the total number of $(2k'+1)$-cycles in $G$ is at most $\Order(n^{k'+1/2})$, we obtain the claimed bound on the expected number of candidate pairs $(x, z)$:
\begin{equation*}
    \sum_{k'=1}^k \Order\left(\frac{n^{k'+1/2}}{s^{k'+1} t^{k'}}\right) \leq \Order\left(\frac{n^{k+1/2}}{s^{k+1} t^k}\right),
\end{equation*}
where the last inequality holds by $s, t \leq n^{1/2}$.
\end{proof}

\begin{lemma}[Running Time of \cref{alg:3sum-to-distance-oracle}]
With high probability, \cref{alg:3sum-to-distance-oracle} runs in expected time
\begin{equation*}
    \widetilde\Order\left(s^2 t \cdot \left(\left(\frac{n^{3/2}}{s t}\right)^{1+p} + \frac{n^{3/2}}{s^2} \cdot \left(\frac{n^{3/2}}{st}\right)^q + \frac{n^{k+1}}{s^{k+1}t^{k+1}}\right)\right).
\end{equation*}
Moreover, if $kp + (k+1)q < 1$ then we can optimize $s$ and $t$ such that the time bound becomes truly subquadratic.
\end{lemma}
\begin{proof}
We can construct the partitions $X_1, \dots, X_s$, $Y_1, \dots, Y_t$ and $Z_1, \dots, Z_s$ in time $\Order(n)$ and prepare the graphs $G_{i, j, \ell}$ in time $\Order(n^{3/2})$ by a single pass over the edge set.

The algorithm runs for $s^2 t$ iterations of the outer loop; focus on one such iteration~$i, j, \ell$. Preprocessing~$G_{i, j, \ell}$ with the distance oracle takes time $\widetilde\Order((n^{3/2} / st)^{1+p})$. Then we issue~$\Order(n^{3/2} / s^2)$ queries, each running in time $\widetilde\Order((n^{3/2} / st)^q)$. By \cref{lem:few-short-paths}, there are at most~$\Order(n^{k+1/2} s^{-k-1} t^{-k})$ candidate pairs in expectation, and only for those we pass the condition \cref{alg:3sum-to-distance-oracle:line:close-condition}. Executing the inner-most loop in \crefrange{alg:3sum-to-distance-oracle:line:loop-test}{alg:3sum-to-distance-oracle:line:report} takes time proportional to the degre of $x$ in $G_{i, j, \ell}$, that is, time $\Order(n^{1/2} / t)$. Summing all contributions, the expected running time becomes:
\begin{equation*}
    \widetilde\Order\left(s^2 t \cdot \left(\left(\frac{n^{3/2}}{s t}\right)^{1+p} + \frac{n^{3/2}}{s^2} \cdot \left(\frac{n^{3/2}}{st}\right)^q + \frac{n^{k+1}}{s^{k+1}t^{k+1}}\right)\right).
\end{equation*}

We now prove that if $kp + (k+1)q < 1$, then the running time becomes subquadratic for some appropriate choice of $s$ and $t$. Let $\epsilon > 0$ be a small constant to be specified later, and set
\begin{align*}
    s &= n^{1/2-\frac{p}{2-2p-2q} - \epsilon}, \\
    t &= n^{1/2-\frac{q}{2-2p-2q} - \epsilon}.
\end{align*}
We analyze the three contributions of the running time in isolation. The first term (i.e., the contribution of the preprocessing time) is
\begin{equation*}
    \widetilde\Order(n^{2-\frac p{2-2p-2q} - \epsilon+p(1/2+\frac{p+q}{2-2p-2q}+\epsilon)}) = \widetilde\Order(n^{2-\frac p{2-2p-2q} - \epsilon+p(\frac1{2-2p-2q}+\epsilon)}) = \widetilde\Order(n^{2-\epsilon(1-p)}).
\end{equation*}
This is subquadratic for any choice of $\epsilon > 0$ as we assume that $p < 1$. The second term (i.e., the contribution of the query time) similarly becomes subquadratic:
\begin{equation*}
    \widetilde\Order(n^{2-\frac q{2-2p-2q}-\epsilon+q(1/2+\frac{p+q}{2-2p-2q}+\epsilon)}) = \widetilde\Order(n^{2-\frac q{2-2p-2q}-\epsilon+q(\frac1{2-2p-2q}+\epsilon)}) = \widetilde\Order(n^{2-\epsilon(1-q)}).
\end{equation*}
For the third term (i.e., the contribution of testing all candidate pairs) we obtain the following bound:
\begin{gather*}
    \widetilde\Order(n^{3/2+\frac{p(k-1)}{2-2p-2q}+\frac{q k}{2-2p-2q}+\epsilon(2k-1)}) = \widetilde\Order(n^{3/2+\frac{kp + (k-1)q - p - q}{2-2p-2q}+\epsilon(2k-1)}).
\end{gather*}
By the same assumption that $kp + (k+1)q < 1$, the exponent becomes strictly smaller than $2$ when ignoring the contribution of $\epsilon$. Therefore, a sufficiently small choice of~$\epsilon > 0$ achieves truly subquadratic running time.
\end{proof}

\subsection{Stretch \texorpdfstring{\boldmath$2 \leq \alpha < 3$}{2 < α ≤ 3}} \label{sec:distance-oracles:sec:small}
In this section we prove the following theorem:

\thmhardnessdistanceoraclesmall*

We use a powerful gadget which was already used in the conditional space lower bounds by Pătraşcu, Roditti and Thorup~\cite{PatrascuRT12}: \emph{Butterfly graphs}. We first define the butterfly graph and then give quick proofs for the properties relevant for our reduction.

\begin{definition}[Butterfly Graph] \label{def:butterfly}
The butterfly graph with alphabet $\sigma$ and dimension $d$ is the $(d+1)$-partite graph with vertex sets $[\sigma]^d \times [d+1]$, and edges
\begin{equation*}
    \Big\{\,(\,(s, i), (s', i+1)\,) : \text{$s[j] = s'[j]$ for all $j \in [d], j \neq i$}\,\Big\}.
\end{equation*}
That is, two vertices $(s, i)$ and $(s', i+1)$ are connected by an edge if and only if the length-$d$ string~$s$ equals~$s'$ in all positions except $i$ (where it might be or might not be equal). We call the vertices $(s, i)$ the \emph{$i$-th layer}, and we occasionally call the 1-st layer the \emph{left layer} and the~$(d+1)$-st layer the \emph{right layer}.
\end{definition}

In particular, we remark that the butterfly graph with alphabet $\sigma$ and dimension $d$ has~$(d+1) \sigma^d$ vertices and $d \sigma^{d+1}$ edges.

\begin{lemma}[Butterfly Graph] \label{lem:butterfly}
Focus on the butterfly graph with alphabet $\sigma$ and dimension~$d$. Then:
\begin{itemize}
\item\emph{Left to right:} The distance from any vertex in the left layer to any vertex in the right layer is exactly $d$.
\item\emph{Left to left:} The probability that two \emph{random} vertices in the left layer have distance at most $2d - 2\ell$ is at most $\sigma^{-\ell}$.
\end{itemize} 
\end{lemma}
\begin{proof}
Observe that in the butterfly graph, exactly the edges from the $i$-th to $(i+1)$-st layer can change the $i$-th position of the strings. This makes the first property obvious: For any two vertices $(s, 1)$ and $(s', d+1)$, follow the unique path which corrects the mismatches between $s$ and $s'$ in positions $1, 2, \dots, d$.

For the second property, let $(s, 1)$ and $(s', 1)$ be two random vertices in the left layer, i.e., let $s$ and $s'$ be random strings in $[\sigma]^d$. The distance between $(s, 1)$ and $(s', 1)$ is exactly two times the largest $i$ for which $s[i] \neq s'[i]$, as we have to reach the $i$-th layer in the butterfly in order to change $s[i]$ into $s'[i]$. Hence, they have distance at most $2d - 2\ell$ only if $s$ equals~$s'$ in the last $\ell$ positions. Since $s$ and $s'$ are random strings, this happens with probability at most $\sigma^{-\ell}$.
\end{proof}

\begin{algorithm}[t]
\caption{The reduction from listing triangles in a $\Theta(r)$-regular tripartite graph $G = (X, Y, Z, E)$ to approximate distance oracles with stretch $\alpha$.} \label{alg:3sum-to-distance-oracle-small}
\begin{algorithmic}[1]
\State Randomly split $Y$ into $Y_1, \dots, Y_t$ \label{alg:3sum-to-distance-oracle-small:line:split}
\ForEach{$j \in [t]$} \label{alg:3sum-to-distance-oracle-small:line:loop-split}
    \State Let $G_j$ be the following graph: Add the vertices $X_i$ and $Z_\ell$, and add a copy of the 
    \Statex[1] butterfly graph with alphabet $\sigma$ and dimension $d$ for each vertex in $Y_j$. For each
    \Statex[1] $(x, y) \in (X \times Y_j) \cap E$, add an edge from $x$ to a random vertex in the left layer of the
    \Statex[1] butterfly graph corresponding to $y$, and similarly for each $(y, z) \in (Y_j \times Z) \cap E$, add
    \Statex[1] an edge from $z$ to a random vertex in the right layer of the butterfly graph
    \Statex[1] corresponding to $y$ \label{alg:3sum-to-distance-oracle-small:line:construct}
    \State Preprocess $G_j$ with the approximate distance oracle \label{alg:3sum-to-distance-oracle-small:line:preprocess}
    \ForEach{$(x, z) \in (X \times Z) \cap E$} \label{alg:3sum-to-distance-oracle-small:line:loop-query}
        \State Query the distance oracle to obtain an estimate $d(x, z) \leq \widetilde d(x, z) \leq \alpha \cdot d(x, z)$ \label{alg:3sum-to-distance-oracle-small:line:query}
        \If{$\widetilde d(x, z) \leq \alpha \cdot (d + 2)$} \label{alg:3sum-to-distance-oracle-small:line:close-condition}
            \ForEach{$y \in Y_j$ with $(x, y) \in E$} \label{alg:3sum-to-distance-oracle-small:line:loop-test}
                \If{$(y, z) \in E$} \label{alg:3sum-to-distance-oracle-small:line:test}
                    \State Report the triangle $(x, y, z)$ \label{alg:3sum-to-distance-oracle-small:line:report}
                \EndIf
            \EndForEach
        \EndIf
    \EndForEach
\EndForEach
\end{algorithmic}
\end{algorithm}

With this gadget in mind, we are ready to state the reduction, see \cref{alg:3sum-to-distance-oracle-small}. Let~\makebox{$d = \ceil{\max(32 / \epsilon, \frac{4}{3-\alpha})}$}. Using \cref{lem:3sum-to-triangle-degree} we start from a $\Theta(r)$-regular $n$-vertex graph (for some parameter $r \geq n^{1/2}$ to be fixed later) which contains at most $\Order(n^{k/2})$ $k$-cycles, for all $4 \leq k \leq 6d+1$, and will list $\Order(nr)$ triangles in time~$\Order(n r^{2-\delta})$. Let $t \leq r^{1-\Omega(1)}$ be another parameter, and let $\sigma = r^{1/d}$.

The reduction is very similar to the one in the previous section, except that we only split the vertex set~$Y$ (in the language of the previous section we have $s = 1$) and that we construct the graphs $G_j$ differently: The difference is that we replace every vertex in $Y_j$ by a copy of the butterfly gadget. The edges from $X$ are connected to a random vertex in the left layer, and the edges from $Z$ are connected to a random vertex in the right layer. Notice that thereby two vertices $x, z$ which are connected by a 2-path via some vertex $y$ in the original graph, are now connected via a~$(d+2)$-path which traverses the butterfly gadget from left to right.

We preprocess each graph $G_j$ with the distance oracle, and query the distance oracle for all edges $(x, z) \in (X \times Z) \cap E$ to obtain a distance estimate $d(x, z) \leq \widetilde d(x, z) \leq \alpha d(x, z)$. We say that a pair $(x, z)$ is a \emph{candidate} pair if the $\widetilde d(x, z) \leq \alpha (d + 2)$. Note that only candidate pairs can be part of a triangle, and we therefore enumerate all candidate pairs $(x, z)$ and all neighbors $y \in Y_j$ of $x$ and test whether $(x, y, z)$ forms a triangle.

We start to analyze the size of the graphs $G_j$. Note that we have to take care of the additional vertices and edges added by the butterfly gadgets.

\begin{lemma}[Size of $G_j$]
With high probability, the following bounds hold for all~$j \in [t]$: The graph $G_j$ has $\Order(n + nr / t)$ vertices and $\Order(n r^{1 + 1/d} / t)$ edges, and the degree of any vertex~$x \in X$ is bounded by $\Order(r/t)$.
\end{lemma}
\begin{proof}
For the degree bound the butterfly gadgets play no role and the proof is the same as in the last section using Chernoff's bound. For the number of vertices, first recall that with high probability there are $\Order(n/t)$ vertices in $Y_j$. Since each vertex in $Y_j$ is replaced by a butterfly graph of size $\Order(d \sigma^d) = \Order(r)$, the bound on the vertices is correct. Moreover, each butterfly graph contributes~$d \sigma^{d+1} = \Order(r^{1+1/d})$ additional edges and therefore also the bound on the edges is as claimed.
\end{proof}

\begin{lemma}[Few Candidates]
Fix $j \in [t]$. In expectation, the expected number of candidate pairs $(x, z) \in (X \times Z) \cap E$ is at most
\begin{equation*}
    \Order\left(\frac{r^{3d+\frac{3+\alpha}2+\frac 4d}}{t^{3d}}\right).
\end{equation*}
\end{lemma}
\begin{proof}
First note that there is a natural correspondence between paths from $x$ to $z$ in the original graph (which we will call \emph{original paths}) and paths in the constructed graph $G_j$ which take the shortest route through the butterfly gadgets (which we will call \emph{inflated paths}).

Observe that any inflated path from $x$ to $z$ of length at most $\alpha \cdot (d+2)$ must be separable into a path which zigzags between $X$ and the butterfly gadgets, followed by a path which zigzags between the butterfly gadgets and $Z$. Any other inflated path would pass through at least three butterfly gadgets (once by traveling from~$X$ to~$Z$, once by traveling back to $X$ and once more by traveling to $Z$ to reach the final destination $z$) which would require length $3d + 2$. Since we set $d > \frac{4}{3-\alpha}$, we have the inequality~\makebox{$\alpha \cdot (d + 2) < 3d + 2$} which leads to a contradiction.

Any such inflated path originates from a $2k$-path in the original graph $G$ (for some $k \leq 3d$) that first zigzags between $X$ and $Y_j$, and then zigzags between $Y_j$ and $Z$. (In particular, for exactly $k-1$ times the path reaches the vertex part $Y_j$ without crossing to the other side from $X$ to $Z$ or vice versa.) Since the edge from~$x$ to $z$ is also present by assumption, this closes a cycle of length $2k+1$ in the original graph.

As we have a good bound on the number of such cycles (namely, there are~$\Order(r^{2k+1})$ many), our strategy is to prove that each cycle becomes a short inflated path only with small probability. First of all, any original $2k$-path as the one described survives only with probability at most $t^{-k}$ in the induced graph $G[X \cup Y_j \cup Z]$. But even if a path survives, we claim that it leads to a short inflated path only with small probability. Since the path has to traverse $k-1$ butterfly gadgets from left to left (or right to right), we expect the path to have length $2k + d + (k-1) \cdot 2d$ ($2k$ steps in the original path plus $d$ steps to cross through one butterfly gadget plus $(k-1) \cdot 2d$ because of the remaining butterfly gadgets). Using \cref{lem:butterfly}, the probability that it has length at most $2k + d + (k-1) \cdot 2d - 2\ell$ is therefore at most
\begin{equation*}
    \sum_{\substack{\ell_1, \dots, \ell_{k-1} \in \Int\\\ell_1 + \dots + \ell_{k-1} = \ell}} \sigma^{-\ell_1} \cdot \ldots \cdot \sigma^{-\ell_{k-1}} = \sum_{\substack{\ell_1, \dots, \ell_{k-1} \in \Int\\\ell_1 + \dots + \ell_{k-1} = \ell}} \sigma^{-\ell} = \Order(\sigma^{-\ell}).
\end{equation*}
Here we hide in the $O$-notation a constant which only depends on $k$ and $\ell$, both of which are functions of $d$ and thereby constants for us.

Hence, for $\ell = \floor{k + \frac d2 + (k-1)d - \frac{\alpha \cdot (d+2)}{2}}$ the probability that the inflated path has length at most $\alpha \cdot (d + 2) \leq 2k + d + (k-1) \cdot 2d - 2\ell$ is at most
\begin{gather*}
    \Order(\sigma^{-\ell}) \\
    \qquad\leq \Order(\sigma^{-k - \frac d2 - (k-1)d + \frac{\alpha \cdot(d+2)}{2} + 1}) \\ \qquad\leq \Order(\sigma^{-d(\frac12+k-1-\frac{\alpha}{2}) + 4}) \\
    \qquad\leq \Order(\sigma^{-d(k-\frac{\alpha+1}{2}) + 4}) \\
    \qquad\leq \Order(r^{-k+\frac{\alpha+1}{2}+\frac{4}d}).
\end{gather*}
By combining the arguments from the previous paragraphs, we obtain that each $2k+1$ cycle survives only with probability $t^{-k}$ and (independently) becomes a short inflated path with probability at most $\Order(r^{-k+\frac{\alpha+1}{2}+\frac{4}d})$. Since the total number of $(2k+1)$-cycles in the original graph is $\Order(r^{2k+1/2})$, we obtain the claimed bound on the expected number of candidate pairs:
\begin{equation*}
    \Order\left(\sum_{k=2}^{3d} \frac{r^{2k+1} r^{-k+\frac{\alpha+1}2+\frac 4d}}{t^{k}}\right) = \Order\left(\sum_{k=2}^{3d} \frac{r^{k+\frac{3+\alpha}2+\frac 4d}}{t^{k}}\right) \leq \Order\left(\frac{r^{3d+\frac{3+\alpha}2+\frac 4d}}{t^{3d}}\right).
\end{equation*}
For the last inequality we have used that $t \leq r$.
\end{proof}

\begin{proof}[Proof of \cref{thm:hardness-distance-oracles-small}]
We pick $r = n^{\frac{2}{1+\alpha}\cdot (1-\delta)}$ and $t = r^{1-\gamma}$, for some $\gamma, \delta > 0$ to be picked later. Recall that we set $d = \ceil{\max(32 / \epsilon, \frac{4}{3-\alpha})}$ and $\sigma = r^{1/d}$. The correctness proof should be clear from the in-text explanations. It remains to analyze the running time with respect to this choice of parameters. Recall that we aim for a running time of the form $(n r^2)^{1-\Omega(1)}$.

First, consider the contribution of querying the distance oracle: We issue $\Order(t n r)$ queries, each of which runs in subpolynomial time, thus amounting for $\Order(n r^{2-\gamma+\order(1)})$. Next, consider the contribution of explicitly testing whether an edge~$(x, z)$ is part of a triangle, that is, the running time of the inner-most loop \cref{alg:3sum-to-distance-oracle-small:line:loop-test}. By the previous lemma we pass the condition in \cref{alg:3sum-to-distance-oracle-small:line:close-condition} at most
\begin{equation*}
    t \cdot \Order\left(\frac{r^{3d+\frac{3+\alpha}2+\frac 4d}}{t^{3d}}\right)
\end{equation*}
times and each call runs in time $\Order(r/t)$. Therefore, the total time for this step becomes
\begin{gather*}
    t \cdot \Order\left(\frac{r^{3d+\frac{3+\alpha}2+\frac 4d}}{t^{3d}} \cdot \frac rt\right) \\
    \qquad= \Order\left(r^{\frac{1+\alpha}{2}} \cdot r^{2+3d\gamma+\frac4d}\right) \\
    \qquad= \Order\left(n^{1-\delta} r^{2+3d\gamma+\frac4d}\right) \\
    \qquad= \Order\left(n r^{2+3d\gamma+\frac4d-\delta}\right).
\end{gather*}
Finally, we need to consider the preprocessing time of the distance oracles. Recall that each graph $G_j$ has $\Order(n r / t)$ vertices and $\Order(n r^{1+1/d} / t)$ edges. Assuming that the preprocesing time of the distance oracle is $\Order(m^{1+\frac{2}{1+\alpha}-\epsilon})$ as in the theorem statement, the total preprocessing time is bounded by
\begin{gather*}
    t \cdot \Order\left(\left(\frac{nr^{1+1/d}}{t}\right)^{1+\frac{2}{1+\alpha}-\epsilon}\right) \\
    \qquad\leq \Order(r (nr^{1/d+\gamma})^{1+\frac{2}{1+\alpha}-\epsilon}) \\
    \qquad\leq \Order(r^{1+(\frac1d+\gamma) \cdot (1+\frac{2}{1+\alpha}-\epsilon)} n^{1+\frac{2}{1+\alpha}-\epsilon}) \\
    \qquad\leq \Order(r^{1+\frac2d+2\gamma} r^{\frac{1}{1-\delta}} n^{1-\epsilon}) \\
    \qquad\leq \Order(r^{2+\frac2d+2\gamma+2\delta-\epsilon} n).
\end{gather*}
We pick $\delta = \epsilon/4$, $d = \ceil{\max(32 / \epsilon, \frac{4}{3-\alpha})}$ (as announced before), and let $\gamma > 0$ be tiny enough. Then both contributions to the running time become $\Order(n r^{2-\Omega(1)})$. This contradicts the 3-SUM hypothesis by \cref{lem:3sum-to-triangle-degree}.
\end{proof}

\subsection{Dynamic Distance Oracles} \label{sec:distance-oracles:sec:dynamic}
In contrast to the previous sections, we now consider \emph{dynamic} distance oracles. That is, we expect the distance oracle to compute distance estimates while the graph undergoes edge insertions and deletions. 
\thmhardnessdistanceoracledynamic*

\begin{algorithm}[t]
\caption{The reduction from listing triangles in a $\Theta(n^{1/2})$-regular tripartite graph $G = (X, Y, Z, E)$ to dynamic approximate distance oracles with stretch $2k-1$.} \label{alg:3sum-to-distance-oracle-dynamic}
\begin{algorithmic}[1]
\State Randomly split $Y$, $Z$ into $Y_1, \dots, Y_t$, $Z_1, \dots, Z_s$ \label{alg:3sum-to-distance-oracle-dynamic:line:split}
\ForEach{$(j, \ell) \in [t] \times [s]$} \label{alg:3sum-to-distance-oracle-dynamic:line:loop-split}
    \State Let $G_{j, \ell}$ be the subgraph of $G$ induced by $Y_j \cup Z_\ell$, where we subdivide each edge into
    \Statex[1] a path of length $10k$ (equivalently, think of this path as an edge of \emph{weight} $10k$), and
    \Statex[1] add an isolated vertex $v$ \label{alg:3sum-to-distance-oracle-dynamic:line:construct}
    \State Preprocess $G_{j, \ell}$ with the dynamic approximate distance oracle (i.e., add the edges
    \Statex[1] one by one)\label{alg:3sum-to-distance-oracle-dynamic:line:preprocess}
    \ForEach{$x \in X$} \label{alg:3sum-to-distance-oracle-dynamic:line:loop-query}
        \State Add an edge $(v, y)$ for each neighbor $y \in Y_j$ of $x$ \label{alg:3sum-to-distance-oracle-dynamic:line:insert}
        \ForEach{$z \in Z_\ell$ with $(x, z) \in E$}
            \State Query the distance $d(v, z) \leq \widetilde d(v, z) \leq (2k-1) \cdot d(v, z)$
            \If{$\widetilde d(v, z) \leq (2k-1) \cdot (10k + 1)$} \label{alg:3sum-to-distance-oracle-dynamic:line:close-condition}
                \ForEach{$y \in Y_j$ with $(x, y) \in E$} \label{alg:3sum-to-distance-oracle-dynamic:line:loop-test}
                    \If{$(y, z) \in E$}
                        \State Report the triangle $(x, y, z)$
                    \EndIf
                \EndForEach
            \EndIf
        \EndForEach
        \State Delete all edges incident to $v$ \label{alg:3sum-to-distance-oracle-dynamic:line:delete}
    \EndForEach
\EndForEach
\end{algorithmic}
\end{algorithm}

\noindent
We again prove the theorem by a reduction from listing $\Order(n^{3/2})$ triangles in a $\Theta(n^{1/2})$-regular $n$-vertex graph $G = (X, Y, Z, E)$ which contains at most $\Order(n^{k'/2})$ $k'$-cycles for all~\makebox{$4 \leq k' \leq 2k+1$} (that is, we use the conditional hardness result from \cref{lem:3sum-to-triangle} with parameter $k_{\max} = 2k+1$).

Let $s, t \leq n^{1/2-\Omega(1)}$ be two parameters. The reduction is given in \cref{alg:3sum-to-distance-oracle-dynamic}. Our analysis is very similar to the analysis in the previous two sections, and we will therefore omit some details. It is easy to prove that the algorithm reports all triangles in $G$ and is therefore correct. The critical part is to analyze the running time. To this end, we first check the size of the graphs $G_{j, \ell}$. Note that the number of vertices in $G_{j, \ell}$ is dominated by the vertices edges added to the graph by the subdivision of edges into paths.

\begin{lemma}[Size of $G_{i, j}$]
With high probability the following bounds hold for all $(j, \ell) \in [t] \times [s]$:
\begin{itemize}
\item The graph $G_{j, \ell}$ has $\Order(n^{3/2} / st)$ vertices and edges.
\item The degree of any vertex $z \in Z_\ell$ in $G_{j, \ell}$ is $\Order(n^{1/2}/t)$.
\end{itemize}
\end{lemma}

We call a pair $(x, z)$ a \emph{candidate} pair if, in the $x$-iteration of the loop in \cref{alg:3sum-to-distance-oracle-dynamic:line:loop-query}, the distance estimate~$\widetilde d(v, z)$ satisfies $\widetilde d(v, z) \leq (2k-1) (10k + 1)$. That is, the condition in \cref{alg:3sum-to-distance-oracle-dynamic:line:close-condition} is satisfied only for candidate pairs.

\begin{lemma}[Few Candidates]
Fix $j, \ell \in [t] \times [s]$. In expectation, the number of candidate pairs $(x, z) \in (X \times Z_\ell) \cap E$ is at most
\begin{equation*}
    \Order\left(\frac{n^{k+1/2}}{s^k t^k}\right).
\end{equation*}
\end{lemma}
\begin{proof}
Focus on a candidate pair $(x, z)$. There must be a neighbor $y \in Y_j$ of $X$ (in the original graph) such that $y$ and $z$ have distance $d(y, z) \leq (2k-1)(10k+1)-1$ in $G_{j, \ell}$. A shortest $y$-$z$-path can therefore zigzag at most $2k-1$ times between $Y_j$ and $Z_\ell$, as otherwise it would have length at least $2k \cdot 10k > (2k-1)(10k+1) - 1$.

Therefore, any candidate pair $(x, z)$ is part of a cycle of length at most $2k' + 1 \leq 2k+1$ in the original graph $G$. For fixed $j, \ell$, the probability that any $(2k'+1)$-cycle survives in the induced subgraph $G[X \cup Y_j \cup Z_\ell]$ is at most $s^{-k'} t^{-k'}$. Therefore, using that in $G$ there are at most $\Order(n^{k'+1/2})$ cycles of length $2k'+1$, we obtain the claimed bound on the number of candidate pairs:
\begin{equation*}
    \sum_{k'=1}^{k} \Order\left(\frac{n^{k'+1/2}}{s^{k'} t^{k'}}\right) \leq \Order\left(\frac{n^{k+1/2}}{s^k t^k}\right),
\end{equation*}
where the last inequality holds by $s, t \leq n^{1/2}$.
\end{proof}

\begin{proof}[Proof of \cref{thm:hardness-distance-oracles-dynamic}]
We run the reduction in \cref{alg:3sum-to-distance-oracle-dynamic}. We omit the correctness proof which is similar to the previous sections, and focus on the running time. We set
\begin{align*}
    s &= n^{\frac12-\frac{u}{2-2u-2q}-\gamma}, \\
    t &= n^{1/2-\frac{q}{2-2u-2q}-\gamma},
\end{align*}
for some small $\gamma > 0$ to be determined later. There are three major contributions to the running time.

First, the time to preprocess the graphs $G_{j, \ell}$ (via adding all edges one by one) is bounded by $\Order(st \cdot n^{3/2} / st)$ times the time to perform a single update and therefore negligible. The time to perform the edge insertions and deletions in \cref{alg:3sum-to-distance-oracle-dynamic:line:insert,alg:3sum-to-distance-oracle-dynamic:line:delete} is bounded by
\begin{gather*}
    \Order\left(st \cdot n \cdot \frac{n^{1/2}}t \cdot \left(\frac{n^{3/2}}{st}\right)^u\right) \\
    \qquad= \Order(n^{2-\frac{u}{2-2u-2q}-\gamma + u \cdot (\frac12 + \frac{u+q}{2-2u-2q}) + u\gamma}) \\
    \qquad= \Order(n^{2+\frac{-u+u-u^2-uq+u^2+uq}{2-2u-2q} - \gamma(1-u)}) \\
    \qquad= \Order(n^{2-\gamma(1-u)}), 
\end{gather*}
which is subquadratic for an arbitrarily small $\gamma > 0$. Similarly, the total query time can be bounded by
\begin{equation*}
    \Order\left(st \cdot n \cdot \frac{n^{1/2}}s \cdot \left(\frac{n^{3/2}}{st}\right)^q\right) = \Order(n^{2-\gamma(1-q)}).
\end{equation*}
It remains to bound the time spend in the inner-most loop in \cref{alg:3sum-to-distance-oracle-dynamic:line:loop-test}. By the previous lemma we pass the condition in \cref{alg:3sum-to-distance-oracle-dynamic:line:close-condition} at most
\begin{equation*}
    st \cdot \Order\left(\frac{n^{k+1/2}}{s^k t^k}\right)
\end{equation*}
times, and each execution of the loop body takes time $\Order(n^{1/2} / t)$. Therefore, the total time spent in the loop is
\begin{gather*}
    st \cdot \Order\left(\frac{n^{k+1}}{s^k t^{k+1}}\right) \\
    \qquad= \Order(n^{3/2+\frac{(k-1)u+kq}{2-2u-2q}+(2k-1)\gamma}) \\
    \qquad\leq \Order(n^{3/2+\frac{ku+(k+1)q-u-q}{2-2u-2q}+(2k-1)\gamma}).
\end{gather*}
By the assumption that $ku + (k+1)q < 1$, the first terms in the exponent is strictly less than $2$, and therefore we can set $\gamma > 0$ sufficiently small to achieve subquadratic running time.
\end{proof}
\section*{Acknowledgements}

We would like to thank Merav Parter and Sebastian Forster for helpful discussions on the upper bounds. We would also like to thank Seri Khoury and Or Zamir for collaboration on short cycle removal in graphs, and Marvin Künnemann and Karol Węgrzycki for collaboration on another related project, both of which inspired us to work on short cycle removal on numbers.

\bibliographystyle{plainurl}
\bibliography{references}

\appendix
\section{Computing Sumsets} \label{sec:sparse-sumset}

There is a rich body of research on computing sparse convolutions~\cite{ArnoldR15,ChanL15,Nakos20,BringmannFN21,BringmannFN22}, but to the best of our knowledge there is no generalization to groups other than the integers $\Int$ or cyclic groups $\Int / p\Int$. In this section we give generalizations to groups $\Field_p^d$. Our goal is to prove the following two lemmas:

\lemsparsesumset*

\lemsparsesumsetwitness*

For the proofs of the lemmas, we introduce some notation. The \emph{group algebra} $\Int[G]$ of $G$ is the set of all functions $f : G \to \Int$ equipped with a \emph{convolution} operation $f \star g$ defined by
\begin{equation*}
    (f \star g)(c) = \sum_{\substack{a, b \in G\\a + b = c}} f(a) \cdot g(b).
\end{equation*}
It was an active line of research to achieve Fast Fourier Transform algorithms not only for cyclic convolutions, but also for general groups. The following theorem is known since the 90's~\cite{BaumCT91}. In fact, since recently there are even efficient---alas, not near-linear-time---algorithms for general groups~$G$~\cite{Umans19}.

\begin{theorem}[Generalized Fast Fourier Transform] \label{thm:fft}
Let $G = \Field_p^d$. Given two functions $f, g \in \Int[G]$, we can their convolution $f \star g$ in time $\Order(|G| \log |G|)$.
\end{theorem}

We will combine the computation of convolutions with linear hashing. To this end, we prepare the following lemma:

\begin{lemma}[Convolutions and Hashing]
Let $h : G \to G'$ be a linear map, let $f, g \in \Int[\Field_p^d]$ be arbitrary and let $f', g' \in \Int[\Field_p^{d'}]$ be defined by
\begin{align*}
    f'(x) = \sum_{\substack{a \in \Field_p^d\\h(a) = x}} f(a),
    \qquad g'(y) = \sum_{\substack{b \in \Field_p^d\\h(b) = y}} g(b).
\end{align*}
Then:
\begin{align*}
    (f' \star g')(z) = \sum_{\substack{c \in \Field_p^d\\h(c) = z}} (f \star g)(c).
\end{align*}
\end{lemma}
\begin{proof}
The proof is a simple calculation:
\begin{gather*}
    (f' \star g')(z) = \sum_{\substack{x, y \in \Field_p^{d'}\\x + y = z}} f'(x) \cdot g'(y) \\
    \qquad= \sum_{\substack{x, y \in \Field_p^{d'}\\x + y = z}} \left(\sum_{\substack{a \in \Field_p^d\\h(a) = x}} f(a)\right) \cdot \left(\sum_{\substack{b \in \Field_p^d\\h(b) = y}} g(b)\right) \\
    \qquad= \sum_{\substack{a, b \in \Field_p^{d}\\h(a) + h(b) = z}} f(a) \cdot g(b) \\
    \qquad= \sum_{\substack{c \in \Field_p^d\\h(c) = z}} (f \star g)(c). \qedhere
\end{gather*}
\end{proof}

\begin{proof}[Proof of \cref{lem:sparse-sumset}]
As a first step we will show how to obtain a small superset $X \supseteq A + B$, by calling our algorithm recursively: Let $A', B'$ be the sets $A, B$ after chopping off the last coordinate from each vector. We compute $A' + B'$ recursively, and let $X$ be the set of all vectors which are equal to some vector in~$A' + B'$ in the first $d-1$ coordinates, and arbitrary in the last coordinate. We have clearly constructed a superset~$X \supseteq A + B$, and moreover since $|A' + B'| \leq |A + B|$, $X$ has size at most $|A + B| \cdot p$.

Next, we apply a hashing approach: Let $d' = \ceil{\log_p(100 \cdot |X|)}$, let $G' = \Field_p^{d'}$ and let $h : G \to G'$ be a random linear map. The subgroup $G'$ has size $100 |X| \leq |G'| \leq 100p \cdot |X|$. We claim that for any element~$x \in X$, the probability that $x$ is \emph{isolated} under the hashing (that is, that there is no other $y \in X$ with $h(x) = h(y)$) is at least $\frac{9}{10}$. Indeed, the collision probability is \raisebox{0pt}[0pt][0pt]{$\Pr(h(x) = h(y)) \leq |G|^{-1} \leq \frac{1}{100 \cdot |X|}$}, therefore is suffices to take a union bound over all possible $|X|$ elements. Our goal is to test for each isolated element whether it appears in $A + B$ (and further compute its multiplicity $r_{A, B}(x)$)

We let $f : G \to \Int$ be the indicator function of the set $A$ (represented sparsely), and similarly we let~$g : G \to \Int$ be the indicator function of $B$. It is easy to check that $(f \star g)(c)$ is positive if and only if~$c \in A + B$. In fact, we have the stronger property that $(f \star g)(c) = r_{A, B}(c)$. We compute $f'$ and $g'$ as defined in the previous lemma (represented densely) with the hash function $h$, and we compute $f' \star g'$ using \cref{thm:fft}. The previous lemma yields that for every isolated element $x \in X$, we have that $(f' \star g')(h(x)) = (f \star g)(x)$. Our algorithm therefore computes the set of isolated elements (by evaluating the hash function on all inputs~$X$), and for each isolated element $x$ recovers $r_{A, B}(x) \gets (f' \star g')(h(x))$.

As we have argued before, each element is isolated with probability at least $\frac9{10}$. Hence by repeating the process for $\Order(\log n)$ iterations, each element in $X$ was isolated at least once and we have therefore computed~$r_{A, B}(x)$ for all $x \in X$.

The total running time (ignoring the recursive call) can be bounded as follows: Constructing $f, g, f', g'$ is in time $\Order(|A| + |B|)$. Computing $f' \star g'$ takes time $\Order(|G'| \log |G'|)$ using \cref{thm:fft}, and by our choice of $G'$ this becomes $\widetilde\Order(|X|) = \widetilde\Order(p \cdot |A + B|)$. In the same time budget we can also test for each element in~$X$ whether it is isolated under the hashing. In total the running time is $\widetilde\Order(p \cdot |A + B|)$ and the repetitions only add a logarithmic overhead. Note that the recursion reaches depth at most $d$, thereby worsening the running time to $\widetilde\Order(pd \cdot |A + B|)$. 
\end{proof}

\begin{proof}[Proof of \cref{lem:sparse-sumset-witness}]
Focus on some $c \in A + B$. Our goal is to sample a random witness $(a, b) \in A \times B$ with~$a + b = c$ in time $\widetilde\Order(|A + B| \cdot \poly(p, d))$. By repeating $\widetilde\Order(t)$ times we will either produce a list of $t$ distinct witnesses, or if there are less than $t$ witnesses, we have seen every witness at least once.

To sample a witness, we first subsample $A$ and $B$ with rates $1, \frac12, \frac14, \dots$. By a standard isolation argument, with constant probability there is an iteration in which exactly one witness for $c$ survives in the subsets. Moreover, the surviving witness is uniformly distributed among all witnesses. We therefore focus on the goal to recover a witness under the promise that there is a unique witness in the instance. Note that the subsampling incurs only a polylogarithmic overhead.

We will now apply \cref{lem:sparse-sumset} to retrieve the unique witness $(a^*, b^*) \in A \times B$ with $a^* + b^* = c$. Our strategy is to recover $a^*$ entry by entry. (We can then recover $b^*$ via $b^* = c - a^*$.) Focus on some coordinate~$i$. We partition $A$ into $p$ subsets $A_0, \dots, A_{p-1}$ where $A_j$ contains all vectors $a \in A$ which are equal to $j$ in the $i$-th coordinate. We compute $A_0 + B, \dots, A_{p-1} + B$ using \cref{lem:sparse-sumset}. Note that $c$ is in exactly one of these sets, namely the set $A_j + b$ where $j$ is the entry of $a^*$ at coordinate $i$. We have therefore successfully recovered the $i$-th coordinate of $a^*$. By repeating the same algorithm for all $i = 1, \dots, d$, we have successfully recovered $a^*$. We have called the sparse sumset algorithm $p \cdot d$ times, therefore the total time is bounded by $\widetilde\Order(|A + B| \cdot \poly(p, d))$ to find a single witness, and by $\widetilde\Order(t \cdot |A + B| \cdot \poly(p, d))$ to find a list of $t$ witnesses.
\end{proof}
\section{The Balog-Szemerédi-Gowers Theorem} \label{sec:bsg}

Our goal in this section is to prove how the following theorem follows from the work by Chan and Lewenstein.

\thmbsg*

In their paper, they obtain the following result:

\begin{theorem}[Theorem 2.1 and Lemma 7.2 in \cite{ChanL15}] \label{thm:chan-lewenstein}
Let $A,B \subseteq G$ and $E \subseteq A \times B$. Suppose that $|A| \cdot |B| = \Theta(n)^2$, $|\{a+b \mid (a,b) \in E\}| \le t n$, and $|E| \ge \alpha n^2$. Then there exist subsets $A' \subseteq A$ and $B' \subseteq B$ such that:
\begin{enumerate}
\item $|A'+B'| \le O((1/\alpha)^5 t^3 n)$ and
\item $|E \cap (A' \times B')| \ge \Omega(\alpha |A'| |B|) \ge \Omega(\alpha^2 n^2)$.
\end{enumerate}
Given $A,B$ and  query access to $E$, such sets $A',B'$ can be computed by a randomized algorithm in time $\widetilde\Order((1/\alpha)^6(|A|+|B|))$.
\end{theorem}

\begin{proof}[Proof of \cref{thm:bsg}]
Let $C \subseteq A + A$ be the subset containing all elements $x$ with multiplicity $r_{A, A}(x) \geq \frac{|A|}{2K}$. We claim that $|C|$ has size at least $\frac{|A|}{2K}$ as otherwise we would have
\begin{gather*}
    E(A) = \sum_{x \in G} r_{A, A}(x)^2 \\
    \qquad= \sum_{\substack{x \in G\\r_{A, A}(x) \leq |A|/2K}} r_{A, A}(x)^2 + \sum_{\substack{x \in G\\r_{A, A}(x) > |A|/2K}} r_{A, A}(x)^2 \\
    \qquad< \frac{|A|}{2K} \cdot \sum_{\substack{x \in G\\r_{A, A}(x) \leq |A|/2K}} r_{A, A}(x) + \frac{|A|}{2K} \cdot |A|^2 \\
    \qquad= \frac{|A|^3}{K}.
\end{gather*}
Let $C_0 \subseteq C$ be an arbitrary subset of size exactly $\frac{|A|}{2K}$. We will apply \cref{thm:chan-lewenstein} with the bipartite graph with vertex parts $A$ and $B = A$ and edges
\begin{equation*}
    E = \set{(a, b) \in A^2 : a + b \in C_0}.
\end{equation*}
Since each element in $C_0$ contributes at least $\frac{|A|}{2K}$ edges to the graph and since $|C_0| = \frac{|A|}{2K}$, we conclude that~\raisebox{0pt}[0pt][0pt]{$|E| \geq \frac{|A|^2}{4K^2}$}. We can therefore apply \cref{thm:chan-lewenstein} with parameters $n = |A| = |B|$, $\alpha = \frac1{4K^2}$ and~$t = 1$. In this way we obtain subsets $A', B' \subseteq A$, and we claim that the set $A'$ is as desired.

We first check that $A'$ and $B'$ are sufficiently large. \cref{thm:chan-lewenstein} implies that $|A'| \cdot |B'| \geq \Omega(\alpha |A'| n) \geq \Omega(\alpha^2 n^2)$. In particular, it follows that $|A'| \geq \Omega(\alpha n) = \Omega(K^{-2} n)$ and $|B'| \geq \Omega(\alpha n) \geq \Omega(K^{-2} n)$.

To see that $|A' + A'|$ is small, we first note that the theorem implies that $|A' + B'| \leq \Order(\alpha^{-5} t^3 n) = \Order(K^{10} n) = \Order(K^{12} |B'|)$. We apply the Plünnecke-Ruzsa inequality (\cref{lem:pluennecke} with inputs~$A \gets B'$ and~\makebox{$B \gets A'$}) to conclude that $|A' + A'| \leq \Order(K^{24} |B'|) \leq \Order(K^{24} n)$. Finally, the running is bounded by~$\widetilde\Order(K^{12} |A|)$ as claimed.
\end{proof}

\end{document}